\documentclass{article}

\usepackage{fancyhdr}
\usepackage{amsfonts}
\usepackage{amsmath, amsthm, amssymb}
\usepackage {epsfig}
\RequirePackage{bm}

\newtheorem{theorem}{Theorem} 
\newtheorem*{condition}{Condition}
\newtheorem{proposition}{Proposition}
\newtheorem{example}{Example}
\newtheorem{definition}{Definition}

\begin{document}

\title{Dependence Structure of Spatial Extremes Using Threshold Approach}
\author{Soyoung Jeon\\ 
\and
Richard L. Smith}
\date{September 27, 2012}

\maketitle
{\center Department of Statistics \& Operations Research\\University of North Carolina at Chapel Hill\\Chapel Hill, NC 27599\\soyoung@live.unc.edu\\}

\renewcommand{\headrulewidth}{0.0pt}
\thispagestyle{fancy}
\lhead{\it Submitted to Extremes}

\newcommand{\Bdtheta}{\boldsymbol\theta}
\newcommand{\Bdeta}{\boldsymbol\eta}
\newcommand{\var}{\operatorname{Var}}
\newcommand{\cov}{\operatorname{Cov}}

\newcommand{\be}{\begin{equation}}
\newcommand{\ee}{\end{equation}}
\newcommand{\beq}{\begin{eqnarray*}}
\newcommand{\eeq}{\end{eqnarray*}}

\newcommand{\expyx}{\bigg(\frac{\sqrt{\gamma}}{2}+\frac{1}{\sqrt{\gamma}}\log\frac{y}{x}\bigg)}
\newcommand{\expxy}{\bigg(\frac{\sqrt{\gamma}}{2}+\frac{1}{\sqrt{\gamma}}\log\frac{x}{y}\bigg)}
\newcommand{\eexpyx}{\bigg(\frac{1}{4\sqrt\gamma}-\frac{1}{2\sqrt{\gamma^3}}\log\frac{y}{x}\bigg)}
\newcommand{\eexpxy}{\bigg(\frac{1}{4\sqrt\gamma}-\frac{1}{2\sqrt{\gamma^3}}\log\frac{x}{y}\bigg)}

\newcommand{\logyx}{\bigg(\frac{\sqrt{\gamma}}{2}+\frac{y-x}{\sqrt{\gamma}}\bigg)}
\newcommand{\logxy}{\bigg(\frac{\sqrt{\gamma}}{2}+\frac{x-y}{\sqrt{\gamma}}\bigg)}
\newcommand{\llogyx}{\bigg(\frac{1}{4\sqrt\gamma}-\frac{y-x}{2\sqrt{\gamma^3}}\bigg)}
\newcommand{\llogxy}{\bigg(\frac{1}{4\sqrt\gamma}-\frac{x-y}{2\sqrt{\gamma^3}}\bigg)}

\begin{abstract}

\vspace{0.5cm}

The analysis of spatial extremes requires the joint modeling of a spatial process at a large number of stations and max-stable processes have been developed as a class of stochastic processes suitable for studying spatial extremes. Spatial dependence structure in the extreme value analysis can be measured by max-stable processes. However, there have been few works on the threshold approach of max-stable processes.

We propose a threshold version of max-stable process estimation and we apply the pairwise composite likelihood method by \cite{padoan:ribatet:sisson:2009} to estimate spatial dependence parameters. It is of interest to establish limit behavior of the estimates based on the settings of increasing domain asymptotics with stochastic sampling design. Two different types of asymptotic normality are drawn under the second-order regular variation condition for the distribution satisfying the domain of attraction. The theoretical property of dependence parameter estimators in limiting sense is implemented by simulation and a choice of optimal threshold is discussed in this paper.

\end{abstract}

\section{Introduction}
\label{intro}
Extreme value theory and its application are dealing with related methodologies to understand phenomena of rare events such as flooding, high temperatures and precipitations in environmental data. The behavior of rare events requires understanding of the tail distribution.

Extreme value theory has been studied for the univariate case in which extremes are observed as a single variable, during a few decades since \cite{fisher:tippett:1928} and \cite{leadbetter:lindgren:rootzen:1983}. \cite{smith:2003} and \cite{beirlant:et:al:2004} provides statistical methods in the analysis of extremes, and \cite{coles:2001} is a very useful reference with the introduction of modeling and applications of extreme values. Multivariate extreme value theory has been developed to build the modeling of joint extremal behavior. \cite{resnick:1987} reviewed relevant theories in the view of probability and measure theory for multivariate extremes.

In a spatial context, a single quantity (e.g., sea level) is measured at multiple locations and the observed data are spatial variables which are distributed across the earth's surface. Therefore one ultimately requires the modeling of spatial extremes, and a spatial dependence among the different locations is of interest. \cite{Cooley:Jessi:etal:2012} introduces several references with issues of spatial extremes.

It is natural to consider a stochastic process when the sample maxima are observed at each site of a spatial process. Max-stable processes have been developed as an infinite dimensional generalization of multivariate extremes. The first general characterization of max-stable processes was by \cite{haan:1984}, and \cite{smith:1990} has constructed a special case of max-stable processes which provides the useful interpretation of extreme rainfall models. Statistical techniques based on the Smith's max-stable model have been developed by \cite{coles:1993} and \cite{coles:tawn:1996} and the well-known classes of max-stable processes are discussed further by \cite{schlather:2002} and \cite{kabluchko:schlather:dehaan:2009}. However the modeling of max-stable processes did not give a straightforward usage due to the complexity and unavailability of the full likelihood for the max-stable model, and \cite{padoan:ribatet:sisson:2009} developed the maximum composite likelihood approach to fit max-stable processes.

Though max-stable processes for blocked maxima approach are on the exploratory stage, the research on max-stable processes with exceedances over threshold has hardly been considered. In this paper, we are concerned with the development of a threshold approach using max-stable processes in spatial extremes. We review the background of extreme value theory, max-stable processes and spatial dependence measure in Section \ref{sec:1}. In Section \ref{sec:2} we introduce our methodology to model exceedances over threshold using max-stable processes and describe its theoretical framework. Section \ref{sec:3} develops asymptotic properties of spatial dependence parameter estimates, which are illustrated with a simulation study.

\section{Modeling of spatial extremes}
\label{sec:1}

\subsection{Extreme value theory}
Let $X_1,\cdots,X_n$ be i.i.d. random variables with the same probability distribution $F$ and let $M_n = \max(X_1,\cdots,X_n)$ be the maximum. If $M_n$ converges under renormalization to some nondegenerate limit, then the limit must be a member of the parametric family, i.e. there exist suitable normalizing constants $a_n>0$, $b_n$ and the distribution $\widetilde{G}$ such that
\begin{equation} \label{e:three types}
P \bigg\{\frac{ M_n - b_n}{a_n} \leq x\bigg\}=F^n(a_nx+b_n)\longrightarrow \widetilde{G}(x), \quad \mbox{ as } n\rightarrow \infty
\end{equation}
where $\widetilde{G}$ is a nondegenerate distribution function. The distribution function $\widetilde{G}$ which is possible limit laws for maxima of i.i.d. sequences has one of three \textit{Extreme Value Distributions} (EVD).

The three types of EVD can be represented as $G$ combining into a single parametric family distribution, which is called the \textit{Generalized Extreme Value} (GEV) distribution:
\begin{eqnarray*}
G(x;\mu,\psi,\xi)=\exp \bigg\{-\bigg(1+\xi \frac{x-\mu}{\sigma}\bigg )_+ ^{-1/\xi}\bigg\} \mbox{,}
\end{eqnarray*}
where $y_+ =\max (0,y)$, $\mu$ is a location parameter, $\sigma>0$ is a scale parameter and $\xi$ is a shape parameter which determines the tail behavior. The Generalized Extreme Value distribution $G$ has a max-stable property: if $X_1,\cdots,X_N$ are i.i.d. from $G$, then $\max(X_1,\cdots,X_N)$ also has the same distribution, i.e.
\begin{eqnarray*}
G^N(x)=G(A_N x+B_N) \mbox{ for existing constants } A_N>0, B_N.
\end{eqnarray*}
\cite{leadbetter:lindgren:rootzen:1983} showed the relationship between extreme value distributions and max-stable distributions that any extreme value distribution is max-stable and vice versa.

The form of the limiting distribution is invariant under monotonic transformation. Therefore, without loss of generality we can transform the GEV distribution into a specific standard form, called unit Fr\'{e}chet distribution,
\begin{eqnarray*}
P \Bigg\{\bigg(1+\xi \frac{M_n-\mu}{\sigma}\bigg)_+ ^{1/\xi} \leq z\Bigg\}=P(Z \leq z)=\exp (-1/z), \quad z>0,
\end{eqnarray*}
and note that the unit Fr\'{e}chet form is a distribution which has the max-stable property.

Multivariate extreme value theory is concerned with the joint distribution of extremes of two or more random variables. If $G$ is a multivariate EVD, the marginal distribution must be represented by the GEV distribution and each marginal GEV distribution can be transformed into unit Fr\'{e}chet margin, which has the max-stable property.

The finite-dimensional framework of multivariate extreme distribution is extended to an infinite-dimensional generalization with spatial processes. The infinite-dimensional extremes has quite analogous extension to the theory of max-stable random vector. Let $\cal{S}$ be a study region and denote $s$ as a location in the study region. If there exist normalizing sequences $a_n(s)$ and $b_n(s)$ for all $s \in \cal{S}$ such that the sequence of stochastic processes
\begin{equation} \label{e:stocProc}
\max_{i=1,\cdots,n} \frac{X_i(s)-b_n(s)}{a_n(s)} \stackrel {d}{\longrightarrow} Y(s)
\end{equation}
where $Y(s)$ is non-degenerate for all $s$, then the limit process $Y(s)$ is a max-stable process. A finite sample $\{Y(s_1),\cdots,Y(s_D)\}$ can be concerned as a realization of a spatial process $Y(s)$ for more realistic setting.

\subsection{Max-stable process and composite likelihood}
Suppose $X(s), s \in \mathcal{S}$ is a stochastic process, where ${\mathcal{S}}\subseteq {\mathbb{R}}^{d}$ is an arbitrary index set. We can interpret $X(\cdot)$ as a spatial process and an appropriate generalization of multivariate extremes can be made in terms of spacial processes as following: for each $n\geq1$, there exist continuous functions $a_n(s)$ positive and $b_n(s)$ real, for $s\in \mathcal{S}$ such that
\begin{equation} \label{e:max stable}
{Pr}^n\bigg\{\frac{ X(s_j) - b_n(s_j)}{a_n(s_j)} \leq x(s_j), j=1,\cdots,K\bigg\} \longrightarrow G_{s_1,\cdots,s_K}(x(s_1),\cdots,x(s_K)).
\end{equation}
Then $G_{s_1,\cdots,s_K}$ is a multivariate extreme value distribution and the limiting process is \textit{max-stable} if (\ref{e:max stable}) holds for all possible subsets $s_1,\cdots,s_K \in \mathcal{S}$. Note that this is equivalent to the expression in equation (\ref{e:stocProc}).

We are interested in modeling and estimation using max-stable processes for extremes observed at each site of a spatial process. A general representation of max-stable processes was first given by \cite{haan:1984}. The conceptual idea of max-stable processes can be constructed by two components: a stochastic process $W(s)$ and a Poisson process $\Pi$ with intensity $d\zeta/\zeta^2$ on $(0,\infty)$. If $\{W_i(s)\}_{i \in \mathbb{N}}$ is independent copies of $W(s)$ with $E[W(s)]=1$ for all $s$ and $\zeta_i \in \Pi,~i\geq 1,$ is points of the Poisson process, then
\[
Y(s)=\max_{i \geq 1} \zeta_i W_i(s), \quad s \in \cal{S}
\]
is a max-stable process with unit Fr\'{e}chet margins. The construction of different max-stable processes can be differentiated from different choices of the $W(s)$ process and the well-known classes of max-stable processes are discussed by \cite{smith:1990}, \cite{schlather:2002} and \cite{kabluchko:schlather:dehaan:2009}.

Bivariate joint distribution is derived for each max-stable process model. The Smith model has the exact form of bivariate distribution
\begin{align}
\nonumber &P \big( Y(s_1)\leq y_1, Y(s_2)\leq y_2 \big)\\
\label{e:smith} & \quad =\exp\Bigg\{-\frac{1}{y_1}\Phi\Bigg(\frac{a}{2}+\frac{1}{a}\log\frac{y_2}{y_1}\Bigg)
-\frac{1}{y_2}\Phi\Bigg(\frac{a}{2}+\frac{1}{a}\log\frac{y_1}{y_2}\Bigg)\Bigg \}
\end{align}
where $a=\sqrt{(s_1-s_2)^T \Sigma^{-1}(s_1-s_2)}$ characterizing spatial dependence with covariance matrix $\Sigma$ and $\Phi$ is the standard normal cumulative distribution function. Bivariate marginal distribution of Schlather model is given by
\begin{equation}
\label{e:schlather} \exp\Bigg\{-\frac{1}{2}\Bigg(\frac{1}{y_1}+\frac{1}{y_2}\Bigg)\Bigg(1+\sqrt{1-2(\rho(h)+1)\frac{y_1 y_2}{(y_1+y_2)^2}}\Bigg)\Bigg \}
\end{equation}
and the correlation $\rho(h)$ represents spatial dependence where $h$ is the Euclidean distance, $\|s_1-s_2\|$, between two stations. Max-stable model of \cite{kabluchko:schlather:dehaan:2009} is called the Brown-Resnick process and the closed form of the bivariate distributions associated to the variogram $\gamma$ is given by
\begin{equation}
\label{e:BR} \exp\Bigg\{-\frac{1}{y_1}\Phi\Bigg(\frac{\sqrt{\gamma(h)}}{2}+\frac{1}{\sqrt{\gamma(h)}}\log\frac{y_2}{y_1}\Bigg)
-\frac{1}{y_2}\Phi\Bigg(\frac{\sqrt{\gamma(h)}}{2}+\frac{1}{\sqrt{\gamma(h)}}\log\frac{y_1}{y_2}\Bigg)\Bigg \}.
\end{equation}

\textit{Pairwise composite likelihood} We are interested with the analysis of spatial extremes at a large number of stations and the standard methods of estimation, such as MLE and Bayes methods, require a full likelihood. However the full likelihood for the max-stable processes may not be available due to the complexity of its analytic form. With the lack of an explicit form of the joint distribution, \cite{padoan:ribatet:sisson:2009} developed a pairwise composite likelihood approach to fit max-stable processes, based on a composite likelihood method by \cite{lindsay:1988}.

Assume $M$ i.i.d. replications of a stochastic process with bivariate densities $f(y_i,y_j;\mathbf{\psi})$, $1\leq i,j\leq K,$ in a spatial region with $K$ locations. Then the pairwise composite log-likelihood is defined by
\begin{equation} \label{e:padoan}
l_{\mathcal{P}}(\mathbf{\psi};{\mathbf{Y}})=\sum_{m=1}^{M} \sum_{i=1}^{K-1} \sum_{j=i+1}^{K} w_{ij} \log{f(y_{mi},y_{mj};\mathbf{\psi})},
\end{equation}
where $(i,j)$ is a pair of stations and $w_{ij}$ is nonnegative weight functions. One may set the weight as an indicator function, i.e., $w_{ij}=1$ if $\parallel s_1-s_2\parallel \leq \delta$, and $0$ otherwise. The \textit{maximum pairwise composite likelihood estimator} (MCLE), $\hat{\mathbf{\psi}}$, is chosen to maximize (\ref{e:padoan}). \cite{padoan:ribatet:sisson:2009} stated the asymptotic properties of MCLE based on the joint estimation, which maximizes the pairwise composite likelihood instead of the full likelihood.

\subsection{Dependence of spatial extremes: Extremal coefficient}
In the analysis of spatial extremes, one can be interested with measuring spatial dependence among locations and a metric characterizing the tail dependence is \textit{extremal coefficient}. Suppose a $d$-dimensional random variable $\mathbf{X}$ has the common marginal distributions $F(x)$. The extremal coefficient $\theta_d$ can be defined by the relation
\[
Pr\{\max(X_1,\cdots,X_d) \leq x\}=F^{\theta_d}(x).
\]
Assuming the standard form of unit Fr\'{e}chet distribution on each margin, we can characterize the dependence among the components of marginal distribution independently. Let $\mathbf{Z}$ be $d$-dimensional maxima with unit Fr\'{e}chet margins and whose multivariate extreme value distribution is expressed as
\begin{equation} \label{e:Vmeasure}
Pr\{Z_1 \leq z_1,\cdots,Z_d \leq z_d\}=\exp\{-V(z_1,\cdots,z_d)\},
\end{equation}
where the exponent measure $V$ is a homogeneous function of order $-1$. Due to the homogeneity of $V$, the extremal dependence can be measured by $V$ which implies complete dependence if $V(z_1,\cdots,z_d) =\max\big(\frac{1}{z_1},\cdots,\frac{1}{z_d}\big)$ and complete independence if $V(z_1,\cdots,z_d) =\frac{1}{z_1}+\cdots+\frac{1}{z_d}$.

The relationship between the extremal coefficient $\theta_d$ and the exponent measure $V$ is drawn from $\theta_d = V(1,\cdots,1)$, and (\ref{e:Vmeasure}) is expressed in terms of extremal coefficient
\begin{equation} \label{e:extcoef}
Pr\{Z_1 \leq z,\cdots,Z_d \leq z\} = \exp \bigg\{-\frac{\theta_d}{z}\bigg\},
\end{equation}
where $1\leq \theta_d \leq d$ with the lower and upper bounds corresponding to complete dependence and complete independence, respectively.

We consider a pairwise extremal coefficient as a special case of (\ref{e:extcoef}) in the spatial domain. Let $Y(s)$ be a spatial process with unit Fr\'{e}chet margin for all $s\in \cal{S}$ and then extremal dependence between different sites $s$ and $s'$ is obtained by,
\[
Pr\{Y(s) \leq y, Y(s') \leq y\}=\exp \bigg\{-\frac{\theta(s-s')}{y}\bigg\}.
\]
A naive estimator of the pairwise extremal coefficient is proposed by \cite{smith:1990}. \cite{schlather:tawn:2003} investigated theoretical properties of the extremal coefficients and proposed self-consistent estimators of $\theta$ (i.e. estimators that satisfy the properties of extremal coefficients) for the multivariate and spatial case.

\section{Modeling for exceedances over threshold}
\label{sec:2}
Consider the distribution of all observations $X$ over a high threshold $u$ and let $Y=X-u>0$, then
\[
F_u(y)=Pr \{ Y\leq y | Y>0 \}=\frac{F(u+y)-F(u)}{1-F(u)}.
\]
As $u \rightarrow x_0=\sup \{x:F(x)<1\}$, we can find a limit $H$ called \textit{Generalized Pareto Distribution} (GPD)
\begin{equation} \label{e:gpd}
F_u(y) \approx H(y; \sigma_u, \xi)
= 1-\bigg(1+\xi \frac{y}{\sigma_u}\bigg)_+^{-1/\xi}.
\end{equation}
\cite{smith:1987} described the bias versus variance tradeoff in the choice of threshold $u$ of univariate case. If the threshold $u$ increases, the variance of estimators will be high due to small $N$ (number of exceedances) while the estimates are biased due to the poor approximation of $F_u(\cdot)$ by $H(\cdot)$ if $u$ is too small. Thus limit theorems on the threshold approach in the literature are presented as $N\rightarrow \infty$ and $u\equiv u_N \rightarrow x_0$ simultaneously.

\cite{pickands:1975} established the rigorous connection between the classical extreme value theory and the generalized Pareto distribution and proved that the limit of the form (\ref{e:gpd}) exists if and only if there exist normalizing constants and the limiting form of $H$ such that the classical extreme value limit (\ref{e:three types}) holds. Thus the limit result for exceedances over thresholds is equivalent to the limit distribution for maxima in this sense.

As in the univariate case, the threshold method has been developed in bivariate case as well. Let $(x_0, y_0)$ denote the upper endpoint of $F$, where $(x_0,y_0)=\sup \{(x,y): F(x,y)<1\}$, and define the conditional distribution of $(X-u,Y-v)$ given $X>u$ or $Y>v$,
\begin{equation} \label{e:cond'dist}
F_{u,v}(x,y)=\frac{F(u+x,v+y)-F(u,v)}{1-F(u,v)}.
\end{equation}
Then the conditional distribution of bivariate exceedances converges to $H$ where $H$ is a multivariate generalized Pareto distribution by \cite{rootzen:tajvidi:2006}.

In this section we develop an alternative methodology for threshold exceedances using max-stable processes with unit Frech\'{e}t margins. We suggest the modeling of the bivariate threshold exceedances by assuming that the asymptotic distribution holds exactly above a threshold and it leads to a simplified dependence structure for max-stable processes as we characterize the dependence among the components of bivariate marginal distribution in (\ref{e:smith}), (\ref{e:schlather}) and (\ref{e:BR}).

The likelihood representation for this threshold method is developed to fit the model and this has a similar idea by \cite{smith:tawn:coles:1997} which establishes a joint distribution for Markov chains where the bivariate distributions were assumed to be of bivariate extreme value distribution form above a threshold. The censored threshold-based likelihood approach is also available for the modeling of spatio-temporal extremes in \cite{huser:davison:2012}.

\textit{Threshold methodology} Suppose we have annual maxima $\{Y_{t^*s}\}$ at site $s=1,\cdots,D$ in year $t^*=1,\cdots,T^*$. We assume the vectors $\{Y_{t^*s}\}$ are independent for different $t^*$ with joint densities given by a max-stable process, i.e., an explicit expression for its bivariate joint distribution is known and the marginal distributions are unit Fr\'{e}chet for each $t^*$ and $s$. Then the joint bivariate distribution of the annual maxima, $F_{AM}$ is written by
\begin{eqnarray*}
F_{AM}(y_{t^*s},y_{t^*s'};\theta)= Pr\{Y_{t^*s} \leq y_{t^*s},Y_{t^*s'} \leq y_{t^*s'};\theta \},
\end{eqnarray*}
where $\theta$ is the dependence parameter which can be estimated by the max-stable model. Now suppose that the daily data are $\{X_{ts}, t=1,\cdots,T, s=1,\cdots,D \}$ and the joint bivariate distribution function is $F_{DA}(x_{ts},x_{ts'};\theta)$. Assume that the daily data $X_{ts}$ form i.i.d. random processes and the annual maxima are $Y_{t^*s}$. Then the relationship between their bivariate distributions is
\begin{align} \label{e:DAAM}
F_{DA}(x_{ts},x_{ts'};\theta) \nonumber &=Pr\{X_{ts}\leq x_{ts},X_{ts'} \leq x_{ts'};\theta \}\\
&=F_{AM}(x_{ts},x_{ts'};\theta)^{1/M}
\end{align}
where $M$ is the number of days in a year. We can have a closed form for $F_{AM}$ from the max-stable theory and also get an expression for $F_{DA}$ from the above representation.

In practice, we would expect to apply some notion of thresholding. Suppose we fix the threshold $u$ and we assume that the same threshold for all locations for convenience. Then we observe exceedances $\{X_{ts}\}$ such that $X_{ts}>u$. Let $\delta_s=I(X_{ts}>u)$ where $I$ is the indicator function. We can obtain the following joint distribution of $(\delta_{s}, X_{ts}, \delta_{s'}, X_{ts'})$ from four possible regions by including or excluding the interval over threshold $u$,
\begin{eqnarray*}
Pr\{\delta_s=0, \delta_{s'}=0\}&=&F_{DA}(u,u)\\
Pr\{\delta_s=1, \delta_{s'}=0, X_{ts}<x_{ts}\}&=&F_{DA}(x_{ts},u)\\
Pr\{\delta_s=0, \delta_{s'}=1, X_{ts'}<x_{ts'}\}&=&F_{DA}(u,x_{ts'})\\
Pr\{\delta_s=1, \delta_{s'}=1\}&=&F_{DA}(x_{ts},x_{ts'}).
\end{eqnarray*}

We extend the threshold version of max-stable processes and apply the maximum composite likelihood method on it. The likelihood contribution of the pair $(x_{ts},x_{ts'})$ derived from the joint bivariate density can be obtained by
\begin{center}
$ L(X_{ts},X_{ts'};\Bdtheta,\Bdeta) = \left\{ \begin{array}{cl}
F_{DA}(u,u) & \mbox{ if }x_{ts}\leq u, x_{ts'}\leq u,\\
\frac{\partial}{\partial x_{ts}}F_{DA}(x_{ts},u) & \mbox{ if }x_{ts}> u, x_{ts'}\leq u, \\
\frac{\partial}{\partial x_{ts'}}F_{DA}(u,x_{ts'}) &\mbox{ if }x_{ts}\leq u, x_{ts'}> u, \\
\frac{\partial^2}{\partial x_{ts} \partial x_{ts'}}F_{DA}(x_{ts},x_{ts'}) &\mbox{ if }x_{ts}> u, x_{ts'}> u.
\end{array} \right. $
\end{center}
where $\Bdtheta$ is the dependence parameter vector and $\Bdeta$ is a vector of marginal GEV parameter. Combining the above likelihood representation with a pairwise likelihood, we assume $T$ i.i.d. replications of a stochastic process with bivariate densities of the unit Frech\'{e}t margins $L(X_{ts},X_{ts'};\Bdtheta,\Bdeta), 1\leq s,s'\leq D$. Then the pairwise composite log-likelihood for a thresholded process is
\begin{equation} \label{e:mcle_th}
l(\Bdtheta,\Bdeta)=\sum_{t=1}^{T}\sum_{s=1}^{D-1} \sum_{s'=s+1}^D {w_{ss'} \log L(X_{ts},X_{ts'};{\Bdtheta},{\Bdeta})}=\sum_{t=1}^T l_{t}(\Bdtheta,\Bdeta)
\end{equation}
where $l_t(\Bdtheta,\Bdeta)=\sum_{s=1}^{D-1} \sum_{s'=s+1}^D {w_{ss'} \log L(X_{ts},X_{ts'};\Bdtheta,\Bdeta_0)}$, $(s,s')$ is a pair of different stations and $T$ is a number of observations. In practice, the marginal parameter $\Bdeta$ will be estimated but we let $\Bdeta$ be the true value $\Bdeta_0$ to simplify theoretical justification. Thus we fix the marginal GEV parameters $\Bdeta=\Bdeta_0$ and estimate the dependence parameter $\Bdtheta$. A dependence parameter $\Bdtheta$ can be estimated by maximizing the pairwise composite likelihood function (\ref{e:mcle_th}) with the known value $\Bdeta_0$.

Suppose $\mathbf{X}^{(t)}=(X_{ts},X_{ts'})$ and denote the composite score functions by pairwise log-likelihood derivatives as
\begin{align*}
& {\mathrm D}(\Bdtheta;\mathbf{X}^{(t)})=\dfrac{\partial l_t(\Bdtheta,\Bdeta_0)}{\partial\Bdtheta},\\
&{\mathbb{D}}(\Bdtheta;\mathbf{X}^{(1)},\cdots,\mathbf{X}^{(T)})={\mathbb{D}}_{\Bdtheta_0}l(\Bdtheta,\Bdeta_0;\mathbf{X}^{(1)},\cdots,\mathbf{X}^{(T)})
=\sum_{t=1}^T {\mathrm D}(\Bdtheta;\mathbf{X}^{(t)}).
\end{align*}
Then the estimating equations
\begin{eqnarray*}
{\mathbb{D}}(\widehat{\Bdtheta};\mathbf{X}^{(1)},\cdots,\mathbf{X}^{(T)})={\mathbb{D}}_{\Bdtheta} l(\widehat{\Bdtheta},\Bdeta_0;\mathbf{X}^{(1)},\cdots,\mathbf{X}^{(T)})=0.
\end{eqnarray*}
The parameter estimator $\widehat \Bdtheta$ is a root to solve above estimating equations and we now start to describe the theoretical framework with more strict conditions to obtain asymptotic properties of the estimator.

\subsection{Second-order regular variation}
To obtain a limiting distribution of $F_{u,v}$ we assume a strict form of condition, so called the \textit{second-order regular variation condition}, for the distribution satisfying the domain of attraction. The ideas of second-order regular variation have been applied to the statistics of extremes. Asymptotic properties of estimators in univariate extreme value theory have been investigated with the second-order regular variation (see \cite{smith:1987}, \cite{haan:stadtmuller:1996}, and \cite{drees:1998}), and the second-order regular variation condition was studied for bivariate extremes by \cite{haan:ferreira:2006}.

\begin{definition}
A function $f(x)$ is regular varying with index $\tau_1$ if for some $\tau_1 \in \mathbb{R}$,
\[ \lim_{t \rightarrow \infty} \frac{f(tx)}{f(t)}=x^{\tau_1}, \quad x>0. \]
The function $f(x)$ is second-order regular varying with the first order $\tau_1$ and the second order $\tau_2$ if there exists a function $q(t)\rightarrow 0$ as $t\rightarrow \infty$ such that
\[ \lim_{t \rightarrow \infty} \frac{\frac{f(tx)}{f(t)}-x^{\tau_1}}{q(t)}=x^{\tau_2}, \quad x>0. \]
\end{definition}
Just as in the univariate case, the representation of bivariate regular variation exists.

\begin{definition}
A function $f(x,y): \mathbb{R}_+^2 \rightarrow \mathbb{R}_+$ is regular varying of index $\tau$ if
\[ \lim_{t \rightarrow \infty} \frac{f(tx,ty)}{f(t,t)}=r(x,y)\]
where $r(\lambda x,\lambda y)=\lambda^{\tau} r(x,y)$ for some $\lambda>0$.
\end{definition}
See \cite{resnick:2007} for the related discussion of multivariate regular variation.

Suppose that $(X_i,Y_i),~i=1,\cdots,n$, is a sequence of i.i.d. random vectors and $F$ be the common distribution of $(X_i,Y_i)$ with marginal distributions $F_1$ and $F_2$. A distribution function $F$ is said to be in the \textit{domain of attraction} of a distribution function $G$, shortly $F \in D(G)$, if
\begin{equation} \label{e:doaG}
\lim_{n\rightarrow\infty} F^n(a_n x+b_n, c_n y+d_n)=G(x,y), \mbox{ } a_n, c_n>0 \mbox{ and } b_n, d_n \in {\mathbb{R}}
\end{equation}
for all $x$ and $y$. The two marginals of $G(x, \infty)$ and $G(\infty, y)$ are one-dimensional extreme value distributions satisfying
\begin{eqnarray*}
\lim_{n\rightarrow\infty} F_1^n(a_n x+b_n)&=&\exp\{-(1+\xi_1 x)^{-1/\xi_1}\},\\
\lim_{n\rightarrow\infty} F_2^n(c_n y+d_n)&=&\exp\{-(1+\xi_2 y)^{-1/\xi_2}\}
\end{eqnarray*}
where $\xi_1$ and $\xi_2$ are real parameters.

Let $(x_0, y_0)$ denote the upper endpoint of $F(x,y)$ and the conditional distribution of $(X-u,Y-v)$ given $X>u$ or $Y>v$ is defined as in (\ref{e:cond'dist}). The equation (\ref{e:doaG}) by taking logarithms can be expressed as
\begin{equation} \label{e:logT}
\lim_{t\rightarrow\infty} t \big\{ 1-F(a_t x+b_t, c_t y+d_t) \big\}=-\log G(x,y)=:\Phi(x,y)
\end{equation}
and it is checked easily that (\ref{e:logT}) implies that
\begin{align*}
\lim_{t\rightarrow\infty} F_{b_t, d_t}(a_t x, c_t y) &= \lim_{t\rightarrow\infty} \bigg(
1-\frac{t \big\{1-F(a_t x+b_t, c_t y+d_t) \big\}}{t \big\{1-F(b_t, d_t) \big\}} \bigg)\\
&= 1-\frac{-\log G(x,y)}{-\log G(0,0)}=:H(x,y)
\end{align*}
where $H$ is a bivariate generalized Pareto distribution. It has been illustrated that $H$ is a good approximation of $F_{b_t, d_t}$ in the sense that
\begin{eqnarray*}
\lim_{t\rightarrow \infty} \sup_{0<(a_t x,c_t y)<(x_0-b_t,y_0-d_t)} |F_{b_t,d_t}(a_t x,c_t y)-H(x,y)|=0,
\end{eqnarray*}
if and only if $F$ is in the maximum domain of attraction of the corresponding extreme value distribution $G$.

Suppose that the second-order regular variation condition in \cite{haan:ferreira:2006} holds: there exists a positive or negative function $\alpha$ with $\lim_{t\rightarrow\infty} \alpha(t)=0$ and a function $Q$ not a multiple of $\Phi$ such that
\begin{equation} \label{e:RV0}
\lim_{t\rightarrow\infty} \frac{t \big\{ 1-F(a_t x+b_t, c_t y+d_t) \big\}-\Phi(x,y)}{\alpha(t)}=Q(x,y)
\end{equation}
locally uniformly for $(x,y) \in (0,\infty]\times(0,\infty]$. Define $U_i$ as the inverse function of $1/(1-F_i),~i=1, 2$ and it is known that for $x,y>0$,
\begin{align*}
\lim_{t'\rightarrow\infty} \frac{U_1(t'x)-U_1(t')}{a(t')} &=\frac{x^{\gamma_1}-1}{\gamma_1},\\
\lim_{t'\rightarrow\infty} \frac{U_2(t'y)-U_2(t')}{c(t')} &=\frac{y^{\gamma_2}-1}{\gamma_2}.
\end{align*}
For $t\neq t'$, define $a_t$, $b_t$, $c_t$ and $d_t$ such that $a(t')\equiv a_t, U_1(t')\equiv b_t, c(t')\equiv c_t,$ and $U_2(t')\equiv d_t$ respectively. Let
\begin{align*}
x_t &:=\frac{U_1(tx)-b_t}{a_t},\\
y_t &:=\frac{U_2(ty)-d_t}{c_t},
\end{align*}
and we could rewrite the form (\ref{e:logT}) as
\begin{eqnarray*}
\lim_{t\rightarrow\infty} t \big\{ 1-F(U_1(tx), U_2(ty)) \big\}=-\log G \bigg(\frac{x^{\gamma_1}-1}{\gamma_1},\frac{y^{\gamma_2}-1}{\gamma_2} \bigg)=:\Phi_0(x,y).
\end{eqnarray*}
It follows the similar form of the second-order condition (\ref{e:RV0}),
\begin{equation} \label{e:RV1}
\lim_{t\rightarrow\infty} \frac{\frac{1-F\big( U_1(\frac{x}{1-F_1(b_t)}), U_2(\frac{y}{1-F_2(d_t)}) \big)}{1-F(b_t,d_t)}-\frac{\Phi_0(x,y)}{\Phi_0(1,1)}}{\alpha \big(\frac{1}{1-F(b_t,d_t)} \big
)}=Q \bigg( \frac{x^{\gamma_1}-1}{\gamma_1},\frac{y^{\gamma_2}-1}{\gamma_2} \bigg).
\end{equation}
We can rewrite the condition (\ref{e:RV1}) and the following second order condition holds for $F_{b_t, d_t}$.

\begin{condition}
There exists a positive or negative function $A(\cdot)$ such that
\begin{equation} \label{e:RV}
F_{b_t, d_t}(a_t x, c_t y)=H(x,y)+A(t)\Psi(x,y)+R_t(x,y), \mbox{ for all }t \mbox{ and }x, y>0
\end{equation}
\begin{itemize}
\item[(i)] $\Psi\equiv0$, $A(t)=o(1)$ and $R_t(x,y)=o(A(t))$ as $t\rightarrow\infty$, or
\item[(ii)] $\Psi$ is continuous and not a multiple of $H$, $A(t)=o(1)$ and $R_t(x,y)=o(A(t))$ as $t\rightarrow\infty$.
\end{itemize}
\end{condition}

The second order regular variation condition implements the domain of attraction condition as a special asymptotic expansion of the conditional distribution $F_{b_t,d_t}$ near infinity. The asymptotic behavior of tail distribution turns out to depend on how the regular variation condition behaves.

In order to obtain asymptotic properties for $\hat\Bdtheta$, we need to understand the behavior of ${\mathbb{E}}{\mathrm D}(\boldsymbol\theta_0)$ given the second-order regular variation, where ${\mathrm D}$ is the score functions of pairwise composite likelihood. The following defines the statement on how integrals of the score functions behave corresponding to the second-order condition.

\begin{proposition} \label{prop}
Let $g_t(x,y)$ be any measurable function. Suppose $F_{b_t, d_t}$ satisfies the condition with (i) or (ii) with function $A$. Define $f_{b_t, d_t}=\frac{d^2 F_{b_t, d_t}}{dxdy}$, $h(x,y)=\frac{d^2 H(x,y)}{dxdy}$ and $\psi(x,y)=\frac{d^2 \Psi(x,y)}{dxdy}$. If
\begin{align} \label{e:Kinteg}
\bigg| g_t(x,y)\bigg\{\frac{f_{b_t,d_t}(a_t x,c_t y)-h(x,y)}{A(t)}-\psi(x,y) \bigg\} \bigg| \leq K(x,y)
\end{align}
which $K(x,y)$ is integrable, then in case of (i)
\begin{eqnarray*}
\int_E{g_t(x,y) dF_{b_t,d_t} (a_t x,c_t y)} = \int_E{g_t(x,y)dH(x,y)}+O(A(t))
\end{eqnarray*}
and in case of (ii)
\begin{align*}
\int_E{g_t(x,y) dF_{b_t,d_t} (a_t x,c_t y)} &= \int_E{g_t(x,y)dH(x,y)}\\
& \quad +A(t) \int_E{g_t(x,y)d\Psi(x,y)}+o(A(t)).
\end{align*}
\end{proposition}

\begin{proof}
As $t \rightarrow \infty$, we have to prove that
\begin{align*}
\int_{(0,\infty]^2} g_t(x,y)\bigg\{\frac{f_{b_t,d_t}(a_t x,c_t y)-h(x,y)}{A(t)}-\psi(x,y) \bigg\}dxdy \longrightarrow 0.
\end{align*}
and by dominated convergence theorem, it is sufficient to show that
\begin{align*}
\bigg| g_t(x,y)\bigg\{\frac{f_{b_t,d_t}(a_t x,c_t y)-h(x,y)}{A(t)}-\psi(x,y) \bigg\} \bigg| \leq K(x,y).
\end{align*}
where $K(x,y)$ is an integrable function.
\end{proof}

Let $g_t(x,y)$ be the score functions from the pairwise composite likelihood of a max-stable process. Note that $\int {g_t(x,y)dH(x,y)}=0$ since $g_t(x,y)$ is the score function. Limit distribution of estimator for dependence parameter can be determined by Condition, and Proposition \ref{prop} implies that condition (\ref{e:Kinteg}) should be satisfied for the limit behavior.

We end this section with an example to demonstrate how the proposition works. Here we focus on the example with a certain type of $g_t(x,y)$, the score function obtained from the composite likelihood of Brown-Resnick process, and we intend to show that the condition (\ref{e:Kinteg}) holds assuming that $F$ is a bivariate normal distribution.

\begin{example}\label{ex}
\textit{(bivariate normal distribution)} Suppose that $(X,Y)$ are i.i.d. from a bivariate normal distribution $F$ with mean 0, variance 1 and correlation coefficient $\rho$. First we can prove that bivariate normal distribution satisfies (\ref{e:RV}) by $A(t)=\frac{1}{2\log t}$ ($A(t)\rightarrow 0$ as $t\rightarrow \infty$).

Define $A(t)=\frac{1}{2\log t}$ and $\psi(x,y)=-\frac{\rho}{2(1-\rho^2)}$ to satisfy the condition (\ref{e:RV}). Now suppose that $g_t(x,y)=\frac{\partial}{\partial \theta}\log f_{DA}(x,y;\theta)$ where $f_{DA}=\frac{\partial^2 F_{DA}(x,y)}{\partial x \partial y}$. Any max-stable process can be fitted for modeling annual maxima of data and we obtain the score function by our threshold method with the composite likelihood approach. Here we arbitrarily choose the Brown-Resnick process with Gumbel margins to obtain the joint bivariate distribution of annual data, $F_{AM}$, and a joint bivariate distribution of daily data, $F_{DA}(x,y)$, is determined by the relation (\ref{e:DAAM}). With some calculations, the following boundness of the product in (\ref{e:Kinteg}) is of interest:
\begin{align} \label{e:exam_bound}
\nonumber \bigg| & g_t(x,y) \bigg\{\frac{f_{b_t,d_t}(a_t x,c_t y)-h(x,y)}{A(t)}-\psi(x,y) \bigg\} \bigg| \\
& \quad \leq \bigg| g_t(x,y)
\exp\bigg(-\frac{x+y}{1+\rho}\bigg) \bigg\{\frac{b_t}{2} \exp\bigg(-a_t^2 \frac{x^2+y^2-2\rho xy}{2(1-\rho^2)}\bigg)-\frac{1}{(1+\rho)^2}\bigg\} \bigg|,
\end{align}
and we can show that (\ref{e:exam_bound}) is bounded by an integrable function. See the details of proof in Appendix A.
\end{example}

\subsection{Spatial structure and sampling design}
Asymptotic results have been proved for spatial processes which are observed at finitely many locations in the sampling region. Central Limit Theorems for spatial data have been studied on infill domain and increasing domain structure under two types of sampling designs, a class of fixed (regular) lattice and stochastic (irregular) designs, in existing literature. Infill domain structure assumes that the sampling region is bounded and locations of data fill in increasingly and densely, while the sampling region is unbounded in the increasing domain structure. \cite{lahiri:2003} is concerned with more complex spatial structure, called mixed asymptotic structure, as a mixture of infill- and increasing domain assumption. In the mixed asymptotic structure, the sampling region is unbounded and sites fill in densely over the region. Covariance parameters are not always consistently estimable if the spatial domain is bounded (\cite{zhang:2004}), while the same parameters are estimable under the increasing domain structure (\cite{mardia:marshall:1984}). Here we focus on the increasing-domain case under stochastic design based on setting and conditions in \cite{lahiri:2003}. Increasing domain structure takes advantage of dealing with asymptotic properties of estimators easily rather than the infill asymptotic structure. We could take account of more realistic setting under the stochastic sampling design than the fixed lattice design.

Suppose that the stationary random field $\{Z(\bm s);\bm s \in \mathbb{R}^d\}$ is observed at many stations $\bm s$ in the sampling region $R_n$. Under the increasing domain structure, $R_n$ is unbounded with $n$ and there is a minimum distance separating any two sites for all $n$. We assume that the sampling region $R_n$ is inflated by the factor $\lambda_n$ from the set $R_0$, i.e.,
\begin{eqnarray*} R_n=\lambda_n R_0. \end{eqnarray*}
For the stochastic designs of sampling sites, we assume that the sampling sites $\{\bm s_1, \cdots, \bm s_n\}$ are obtained from a random vectors $\{\bm x_1, \cdots, \bm x_n\}$ by
\begin{eqnarray*} \bm s_i=\lambda_n \bm x_i, \quad 1 \leq i \leq n \end{eqnarray*}
where $\bm x_i$ is a sequence of i.i.d. random vectors from a continuous probability density function $f(\bm x)$ and its realization $\{\bm x_1, \cdots, \bm x_n\}$ are in $R_0$. In this stochastic design, the sample size $n$ is determined by the growth rate $\lambda_n$ by the relation $n \sim C \lambda_n^d$.

We now consider our threshold approach. Note that we assume the marginal GEV parameter $\Bdeta$ is known as the simplest case, though we would like to address the case $\Bdeta$ unknown as well. Assuming that $\Bdeta$ is known as $\Bdeta_0$, we can rewrite (\ref{e:mcle_th}) and partial derivatives with the temporal domain fixed, as
\begin{align*}
l(\Bdtheta) &=\sum_{i=1}^{n-1} \sum_{j=i+1}^n \sum_{t=1}^{T} {w_{ij} \log L(X_{ti},X_{tj};{\Bdtheta})}\\
&= \sum_{i<j} w_{ij} \log L_{ij}(\Bdtheta),
\end{align*}
\begin{align*}
\dfrac{\partial l(\Bdtheta)}{\partial\Bdtheta} &= \sum_{i<j} \frac{w_{ij}}{L_{ij}(\Bdtheta)}\cdot \frac{\partial L_{ij}(\Bdtheta)}{\partial\Bdtheta},\\
\dfrac{\partial^2 l(\Bdtheta)}{\partial\Bdtheta \partial\Bdtheta^T} &= \sum_{i<j} \frac{w_{ij}}{L_{ij}^2(\Bdtheta)} \bigg\{ \frac{\partial^2 L_{ij}(\Bdtheta)}{\partial\Bdtheta \partial\Bdtheta^T} \cdot L_{ij}(\Bdtheta)-\frac{\partial L_{ij}(\Bdtheta)}{\partial\Bdtheta}\bigg(\frac{\partial L_{ij}(\Bdtheta)}{\partial\Bdtheta} \bigg)^T \bigg\}
\end{align*}
where $w_{ij}$ is the weight function on the $(i,j)$th pair which does not take any values outside $R_n$, $L_{ij}=F_{ij}(x_i,x_j)I_{\{x_i>u,x_j>u\}}+F_{i}(x_i,u)I_{\{x_i>u,x_j \leq u\}}+F_{j}(u,x_j)I_{\{x_i \leq u,x_j>u\}}+F_{DA}(u,u)I_{\{x_i \leq u,x_j \leq u\}}$, and $F_{ij}=\frac{\partial^2 F_{DA}}{\partial x_i \partial x_j}$. Here $u$ is the threshold, not a fixed constant, which varies as the sample size goes to infinity.

We concentrate on the first term of $L_{ij}$ which is the case that both exceed the threshold. Let us define notations related with the first term by
\begin{align*}
Q_K (\Bdtheta) &=\sum_{i<j}^K w_{ij} \log F_{ij}(\Bdtheta)I_{\{x_i>u,x_j>u\}}\\
\dfrac{\partial Q_K (\Bdtheta)}{\partial\Bdtheta} &= \sum_{i<j}^K \frac{w_{ij}}{F_{ij}(\Bdtheta)}\cdot \frac{\partial F_{ij}(\Bdtheta)}{\partial\Bdtheta}I_{\{x_i>u,x_j>u\}}\\
\dfrac{\partial^2 Q_K (\Bdtheta)}{\partial\Bdtheta \partial\Bdtheta^T} &= \sum_{i<j}^K \frac{w_{ij}}{F_{ij}(\Bdtheta)} \bigg\{ \frac{\partial^2 F_{ij}(\Bdtheta)}{\partial\Bdtheta \partial\Bdtheta^T}-\frac{1}{F_{ij}}\bigg(\frac{\partial F_{ij}(\Bdtheta)}{\partial\Bdtheta}\bigg)\bigg(\frac{\partial F_{ij}(\Bdtheta)}{\partial\Bdtheta} \bigg)^T \bigg\}I_{\{x_i>u,x_j>u\}}
\end{align*}
where $K$ is the number of all combination of pairs.

Next we denote the form of the strong mixing assumption to deal with dependence through pairs. Let $X(\bm s_i)=Z(\bm s_i)I(Z(\bm s_i)>u_n)$ and $\mathcal{F}(\mathcal{G}_k)$ be $\sigma$-field generated by $\big\{\big(X(\bm s_i), X(\bm s_j)\big);\bm s_i, \bm s_j \in \mathcal{G}_k, 1\leq i, j\leq n, k=1, \cdots, K\big\}$. For any two subsets $A$ and $B$ of $\mathbb{R}^d$, the mixing condition is defined by
\begin{eqnarray*}
\tilde{\alpha}(\mathcal{G}_1, \mathcal{G}_2)=\sup \{|P(A\cap B)-P(A)P(B)|: A \in \mathcal{F}(\mathcal{G}_1), B \in \mathcal{F}(\mathcal{G}_2)\}
\end{eqnarray*}
and let
\begin{eqnarray*}
d(\mathcal{G}_1, \mathcal{G}_2)=\inf \{|\bm s-\bm s'|:\bm s \in \mathcal{G}_1, \bm s' \in \mathcal{G}_2\}
\end{eqnarray*}
which is the minimum distance from element of a pair $\mathcal{G}_1$ to element of another pair $\mathcal{G}_2$. Then the strong mixing coefficient is defined as
\begin{eqnarray*}
\alpha(a,b)=\sup\{\tilde{\alpha}(\mathcal{G}_1, \mathcal{G}_2): d(\mathcal{G}_1, \mathcal{G}_2) \geq a, \quad \mathcal{G}_1, \mathcal{G}_2 \in \mathcal{R}_3(b)\}
\end{eqnarray*}
where $\mathcal{R}_3(b)\equiv \{ \cup_{i=1}^3 D_i: \sum_{i=1}^3 |D_i|\leq b\}$, the the collection of all disjoint unions of three cubes $D_1, D_2$ and $D_3$ in $\mathbb{R}^d$, and it specifies the general form of the sets $\mathcal{G}_1$ and $\mathcal{G}_2$ that are bounded. Assume that there exist a nonincreasing function $\alpha_1(\cdot)$ such that $\lim_{a\rightarrow\infty} \alpha_1(a)=0$ and a nondecreasing function $\beta(\cdot)$ satisfying
\begin{eqnarray*}
\alpha(a,b) \leq \alpha_1(a)\beta(b), \quad a,b>0.
\end{eqnarray*}

In our approach, what we are interested in is the bivariate function,
\begin{align} \label{e:Z_k}
\frac{\partial}{\partial \Bdtheta}\log F_{ij}(x_i,x_j;\Bdtheta)I(x_i>u, x_j>u)& \doteq g_k \big(X(\bm s_i), X(\bm s_j)\big) \doteq Z_k(\bm s^k),
\end{align}
where $Z_k$ is obviously different from the original process $Z$. Let $\sigma(\cdot)$ denote the auto covariance function of the process $Z_k$ such that for all $\bm s_i,\bm s_j,\bm{h_1}, \bm{h_2} \in \mathbb{R}^d$,
\begin{eqnarray*}
\sigma(\bm{h})=Cov(Z_k,Z_l)=Cov \big[g_k \big(X(\bm s_i), X(\bm s_j)\big), g_l \big(X(\bm s_i+\bm h_1), X(\bm s_j+\bm h_2)\big)\big].
\end{eqnarray*}
Let $s_{1K}^2=\int \int w_{ij}^2(\lambda_n \mathbf{X}_{ij})f(\bm x_i,\bm x_j)d\bm x_i d\bm x_j=E w_{K}^2(\lambda_n \mathbf{X}_1)$ where $\mathbf{X}_{ij}=(\bm x_i,\bm x_j)$, $M_{k}=\{\sup|w_{K}(\bm h)|;~\bm h \in \mathbb{R}^{2d}\}$ and $\gamma_{1k}^2=\frac{M_k^2}{s_{1K}^2}$ to simplify the notation, .

We will use the following conditions which are similar with (S.1)-(S.5) in \cite{lahiri:2003} to prove the asymptotic distribution of process.
\begin{itemize}
\item[(A$^\prime$1)] $\int \int |\sigma(\bm x)| d\bm x <\infty$
\item[(A$^\prime$2)] Let $R_0$ be a Borel set satisfying $R_0^* \subset R_0 \subset \bar R_0^*$ and $R_0^*$ be an open connected subset of $(-1/2,1/2]^d$. The pdf $f(x)$ is continuous, everywhere positive with support $\bar R_0$, the closure of the set $R_0 \subset R^d$.
\item[(A$^\prime$3)] Suppose that $\bm x_i$ are i.i.d. from $f$ over $R_0$ and $f(\bm x_i,\bm x_j)=f(\bm x_i)f(\bm x_j)$. The joint pdf $f(\bm x_i,\bm x_j)\in [m_f,M_f]$ where $m_f$ and $M_f$ are constants in $(0,\infty)$. 
\item[(A$^\prime$4)] $\frac{\int \int w_{ij}\big(\lambda_n (\bm x_i,\bm x_j)\big)w_{pq}\big(\lambda_n (\bm x_i,\bm x_j)+\bm h\big)f^2(\bm x_i,\bm x_j)d\bm x_i d\bm x_j}{\int \int w_{ij}^2\big(\lambda_n (\bm x_i,\bm x_j)\big)f(\bm x_i, \bm x_j)d\bm x_i d\bm x_j}\rightarrow Q_1(\bm h)~\mbox{for all }i\neq p, j\neq q, \bm h \in \mathbb{R}^{2d}$
\item[(A$^\prime$5)] $\frac{\int \int w_{ij}\big(\lambda_n (\bm x_i,\bm x_j)\big)w_{iq}\big(\lambda_n (\bm x_i,\bm x_j)+({0,\bm h})\big)f^2(\bm x_i,\bm x_j)d\bm x_i d\bm x_j}{\int \int w_{ij}^2\big(\lambda_n (\bm x_i,\bm x_j)\big)f(\bm{x}_i, \bm{x}_j)d\bm x_i d\bm x_j}\rightarrow Q_2({\bm h})~\mbox{for all }i=p, j\neq q, \bm h \in \mathbb{R}^{d}$.
\item[(A$^\prime$6)] $\gamma_{1k}^2=\frac{M_k^2}{s_{1K}^2}=O(K^a)$ for some $a\in [0,1/8)$
\item[(A$^\prime$7)] There exist sequences $\{\lambda_{1n}\}$, $\{\lambda_{2n}\}$ with $\{\lambda_{1n}\} \geq \{\lambda_{2n}\} \geq \log \{\lambda_{n}\}$ such that
\begin{itemize}
\item[(i)] $\gamma_{1k}^2(\log n)^2 \big[\frac{\lambda_{1n}}{\lambda_{n}}+\frac{\lambda_{2n}}{\lambda_{1n}} \big]=o(1)$
\item[(ii)] $\gamma_{1k}^4(\log n)^4 \big(\frac{\lambda_{1n}^d}{\lambda_{n}^d}\big)\sum_{k=1}^{\lambda_{1n}}k^{2d-1}\alpha_1(k)=o(1)$
\item[(iii)] $\frac{\lambda_{1n}^d}{\lambda_{n}^d}\alpha_1(\lambda_{2n})\beta(\lambda_n^d)=o(1)$
\item[(iv)] $\gamma_{1k}^2[\lambda_{1n}^d\alpha_1(\lambda_{2n})+\sum_{k=\lambda_{1n}}^\infty k^{d-1}\alpha_1(k)]\beta(\lambda_{1n}^d)=o(1)$
\end{itemize}
\end{itemize}

\begin{theorem} \label{thm1}
Assume that conditions (A$^\prime$1)-(A$^\prime$7) hold. Suppose that $Z_k(\bm s^k)$ in (\ref{e:Z_k}) is a stationary stochastic process such that $E|Z_k(0)|^{2+\delta}<\infty$ and $\int {t}^{d-1}\alpha_1({t})^{\frac{\delta}{2+\delta}}d{t}<\infty$ for some $\delta>0$. If $n/\lambda_n^d\rightarrow C_1 \in (0,\infty)$ as $n\rightarrow \infty$, then
\begin{align*}
(K s_{1K}^2)^{-1/2} & \sum_{k=1}^K w_K(\bm{s}_k)Z_k(\bm{s}^k)\\
& \stackrel {d}{\longrightarrow} N \bigg(0, \sigma(\bm{0})+C_1\int\sigma\big((0,\bm h)\big)Q_2(\bm{h})d\bm{h}+C_1^2\int\sigma(\bm{h})Q_1(\bm{h})d\bm{h} \bigg).
\end{align*}
\end{theorem}
\begin{proof} Proof of Theorem \ref{thm1} is shown in Appendix C.
\end{proof}

\section{Asymptotic properties of spatial dependence parameter estimates}
\label{sec:3}

\subsection{Asymptotic normality and consistency}
\label{sec:3:1}
We use the following regularity conditions to obtain an asymptotic behavior of estimates of dependence parameters.
\begin{itemize}
\item[(A1)] The support $\boldsymbol\chi$ of the bivariate density function of the data does not depend on $\Bdtheta\in \Theta$ and the parameter space $\Theta$ is an open subset of $\mathbb{R}^p$ with identifiable parametrization.
\item[(A2)] The pairwise composite log likelihood is at least twice continuously differentiable in $\Bdtheta$.
\item[(A3)] (smoothness of composite likelihood) $Q_T(\Bdtheta)$ exists and is continuous and ${\mathbb{H}}(\Bdtheta)$ is also continuous in a neighborhood $\Theta^*$ of $\Bdtheta_0$.
\item[(A4)] For all $\Bdtheta_0 \in \Theta$, there exists an integrable function $M(x,y)$ such that
\[ \sup_{\Bdtheta \in \Theta^*} \bigg| \dfrac{\partial^2 Q_T (\Bdtheta;x,y)}{\partial\theta_i \partial\theta_j} \bigg| \leq M(x,y), ~i,j=1,\cdots,p. \]
\item[(A5)] The third partial derivatives of the composite likelihood are bounded by integrable functions.
\item[(A6)] (equivalent condition of Proposition \ref{prop}) The score function of composite likelihood $D$ satisfies that
\begin{equation*}
\bigg| D(x,y)\bigg\{\frac{f_{b_t,d_t}(a_t x,c_t y)-h(x,y)}{A(t)}-\psi(x,y) \bigg\} \bigg| \leq K(x,y)
\end{equation*}
which $K(x,y)$ is integrable.
\end{itemize}

\begin{theorem} (Asymptotic Normality) Suppose that condition (\ref{e:RV}) with (i) or (ii) is satisfied and conditions of Theorem \ref{thm1} hold. Suppose $N\rightarrow\infty$, $(b_k, d_k)=(b(k)_N,d(k)_N)\rightarrow(x_0,y_0)$, and $A(k_N)=O \big(\frac{1}{\sqrt{N s_{1N}^2}} \big)$. If
\begin{eqnarray*}
\sqrt{N s_{1N}^2}A(k_N)\longrightarrow \lambda\in [0,\infty),
\end{eqnarray*}
and either $\lambda=0$ and (i) holds, then the the solutions of likelihood equations verify
\begin{equation} \label{e:normal1}
\sqrt{N}({s_{1N}^2})^{-1/2}(\widehat{\boldsymbol\theta}-\boldsymbol\theta_0) \stackrel {d}{\longrightarrow} N \big(\mathbf{0},{\mathbb{H}}(\boldsymbol\theta_0)^{-1} {\mathbb{V}}(\boldsymbol\theta_0) {\mathbb{H}}(\boldsymbol\theta_0)^{-1}\big),
\end{equation}
or (ii) holds, then
\begin{equation} \label{e:normal2}
\sqrt{N}({s_{1N}^2})^{-1/2}(\widehat{\boldsymbol\theta}-\boldsymbol\theta_0) \stackrel {d}{\longrightarrow} N \big({\mathbb{H}}(\boldsymbol\theta_0)^{-1} \boldsymbol b,{\mathbb{H}}(\boldsymbol\theta_0)^{-1} {\mathbb{V}}(\boldsymbol\theta_0 \big) {\mathbb{H}}(\boldsymbol\theta_0)^{-1}),
\end{equation}
where ${\mathbb{H}}(\Bdtheta_0)= {\mathbb E}[-{\mathrm D}^{\prime}(\Bdtheta_0)]$, $\boldsymbol b=\lim_{N\rightarrow\infty} {\big(N s_{1N}^2\big)}^{-1/2} {\mathbb E} \big\{\sum_{k=1}^{N} {\mathrm D}(\Bdtheta_0;\bm X^{(k)}) \big\}$ (defined below) and
${\mathbb{V}}(\Bdtheta_0)={\mathbb{E}}[{\mathrm D} (\Bdtheta_0) {\mathrm D} (\Bdtheta_0)^T]$.
\end{theorem}

\begin{proof}
Denote that
\begin{align*}
Q_T (\Bdtheta) &=\sum_{i<j}^K \sum_t w_{ij} \log L_{ij}(\Bdtheta)\\
&= \sum_{i<j}^K \sum_t w_{ij} \log \big\{F_{ij}(x_i,x_j)I_{\{x_i>u,x_j>u\}}+F_{i}(x_i,u)I_{\{x_i>u,x_j \leq u\}}\\
& \qquad +F_{j}(u,x_j)I_{\{x_i \leq u,x_j>u\}}+F_{DA}(u,u)I_{\{x_i \leq u,x_j \leq u\}}\big\} \\
& := \sum_{i<j}^K \sum_t w_{ij} \log \big\{L_{1ij}(\Bdtheta;x_i,x_j)+L_{2ij}(\Bdtheta;x_i)
+L_{3ij}(\Bdtheta;x_j)+L_{4ij}(\Bdtheta)\big\},
\end{align*}
\begin{align*}
\dfrac{\partial Q_T (\Bdtheta)}{\partial\Bdtheta} &=\sum_{i<j}^K \sum_t \frac{w_{ij}}{L_{ij}(\Bdtheta)}\cdot \frac{\partial L_{ij}(\Bdtheta)}{\partial\Bdtheta}\\
&= \sum_{i<j}^K \sum_t \frac{w_{ij}}{L_{ij}(\Bdtheta)}\times\\
& \qquad \frac{\partial}{\partial\Bdtheta}\big\{L_{1ij}(\Bdtheta;x_i,x_j)+L_{2ij}(\Bdtheta;x_i)
+L_{3ij}(\Bdtheta;x_j)+L_{4ij}(\Bdtheta) \big\}\\
&:=\sum_{k=1}^N D\big(\Bdtheta;(X_{ti}^{(k)},X_{tj}^{(k)})\big)\\
& =\sum_{k=1}^N \big\{D_1(\Bdtheta)+D_2(\Bdtheta)+D_3(\Bdtheta)+D_4(\Bdtheta)\big\}
\end{align*}
where
\begin{align*}
& D_1(\Bdtheta)=\sum_{i<j}^K \sum_t \frac{w_{ij}}{L_{ij}(\Bdtheta)}\cdot \frac{\partial}{\partial\Bdtheta}L_{1ij}(\Bdtheta;x_i,x_j),\\
& D_2(\Bdtheta)=\sum_{i<j}^K \sum_t \frac{w_{ij}}{L_{ij}(\Bdtheta)}\cdot \frac{\partial}{\partial\Bdtheta}L_{2ij}(\Bdtheta;x_i),\\
& D_3(\Bdtheta)=\sum_{i<j}^K \sum_t \frac{w_{ij}}{L_{ij}(\Bdtheta)}\cdot \frac{\partial}{\partial\Bdtheta}L_{3ij}(\Bdtheta;x_j),\\
& D_4(\Bdtheta)=\sum_{i<j}^K \sum_t \frac{w_{ij}}{L_{ij}(\Bdtheta)}\cdot \frac{\partial}{\partial\Bdtheta}L_{4ij}(\Bdtheta),
\end{align*}
and $K$ is the number of all combination of pairs. We now consider $N$ the number of exceedances as a primary role in deriving the asymptotic behavior. By notations and condition (A2), we have Taylor expansion about $\Bdtheta_0$ as follows.
\begin{small}
\begin{align*}
0 &= \dfrac{\partial Q_T (\Bdtheta)}{\partial\Bdtheta}\bigg|_{\Bdtheta=\widehat{\Bdtheta}}\\
&= \dfrac{\partial Q_T (\Bdtheta)}{\partial\Bdtheta}\bigg|_{\Bdtheta=\Bdtheta_0}
+(\widehat{\Bdtheta}-\Bdtheta_0)^T \dfrac{\partial^2 Q_T (\Bdtheta)}{\partial\Bdtheta^2}\bigg|_{\Bdtheta=\Bdtheta_0}+\frac{1}{2}(\widehat{\Bdtheta}-\Bdtheta_0)^T \dfrac{\partial^3 Q_T (\Bdtheta)}{\partial\Bdtheta^3}\bigg|_{\Bdtheta=\Bdtheta^*} (\widehat{\Bdtheta}-\Bdtheta_0)\\
&\quad \mbox{where } {\Bdtheta}^*\mbox{ lies between }\widehat{\Bdtheta} \mbox{ and }{\Bdtheta_0}\\
&= \sum_{k=1}^{N} {\mathrm D}(\Bdtheta_0;(X_{ti}^{(k)},X_{tj}^{(k)}))
+(\widehat \Bdtheta -\Bdtheta_0)^T \sum_{k=1}^{N}{\mathrm D}^{\prime}(\Bdtheta_0;(X_{ti}^{(k)},X_{tj}^{(k)}))\\
&\qquad +\frac{1}{2}(\widehat \Bdtheta -\Bdtheta_0)^T \sum_{k=1}^{N}{\mathrm D}^{\prime \prime}(\Bdtheta^*;(X_{ti}^{(k)},X_{tj}^{(k)})) (\widehat \Bdtheta -\Bdtheta_0).
\end{align*}
\end{small}
Then we rewrite the equation as
\begin{align*}
& \dfrac{1}{\sqrt{N}} \sum_{k=1}^{N} {\mathrm D}(\Bdtheta_0;\mathbf{X}^{(k)})\\
&\quad = \Bigg\{-\frac{1}{N}\sum_{k=1}^{N}{\mathrm D}^{\prime}(\Bdtheta_0;\mathbf{X}^{(k)})
-\frac{1}{2}(\widehat \Bdtheta -\Bdtheta_0)^T \frac{1}{N} \sum_{k=1}^{N}{\mathrm D}^{\prime \prime}(\Bdtheta^*;\mathbf{X}^{(k)}) \Bigg\} \sqrt{N}(\widehat \Bdtheta -\Bdtheta_0),
\end{align*}
and then
\begin{align} \label{e:taylor}
\nonumber & \sqrt{\frac{N}{s_{1N}^2}} (\widehat \Bdtheta -\Bdtheta_0)\\
& \quad =\Bigg\{\underbrace{-\frac{1}{N}\sum_{k=1}^{N}{\mathrm D}^{\prime}(\Bdtheta_0)-\frac{1}{2}(\widehat \Bdtheta -\Bdtheta_0)^T \frac{1}{N} \sum_{k=1}^{N}{\mathrm D}^{\prime \prime}(\Bdtheta^*)}_{\mbox{(a)}}\Bigg\}^{-1}
\underbrace{\frac{1}{\sqrt{N s_{1N}^2}} \sum_{k=1}^{N} {\mathrm D}(\Bdtheta_0).}_{\mbox{(b)}}
\end{align}
We establish the following for separate terms in equation (\ref{e:taylor}):
\begin{itemize}
\item[(I)] By the consistency of $\widehat \Bdtheta$ and condition (A5), expectation of the last term in parentheses can be ignored.
Since $\widehat \Bdtheta$ is consistent, $\widehat \Bdtheta \in \Theta^*$ with $P_{\Bdtheta_0}$-probability 1. Let $B \subset \Theta^*$ be a closed ball with the center $\Bdtheta_0$. By the condition (A4),
\[ \sup_{\widehat \Bdtheta \in B} \bigg\| \dfrac{\partial^2 Q_T (\Bdtheta_0;x,y)}{\partial\theta_i \partial\theta_j}
-\dfrac{\partial^2 Q_T (\widehat \Bdtheta;x,y)}{\partial\theta_i \partial\theta_j} \bigg\| \]
is bounded and then, for large $N$,
\[ {\limsup}_N \bigg\| \frac{1}{N}\sum_{k=1}^{N}{\mathrm D}^{\prime}(\Bdtheta_0)
-\frac{1}{N}\sum_{k=1}^{N}{\mathrm D}^{\prime}(\widehat \Bdtheta)\bigg\| \leq \varepsilon \]
in probability (see details in \cite{guyon:1995}). $\frac{1}{N}\sum_{k=1}^{N}{\mathrm D}^{\prime}(\Bdtheta_0)-{\mathbb E}[-{\mathrm D}^{\prime}(\Bdtheta_0)]$ converges to $0$ by the law of large numbers, and hence (a) converges to ${\mathbb E}[-{\mathrm D}^{\prime}(\Bdtheta_0,\Bdeta_0)]$ in probability.
\item[(II)] First consider that $(X_{ti}^{(k)},X_{tj}^{(k)}),~k=1,\cdots,N$ are i.i.d. from exact multivariate GPD distribution $H$.
\begin{align*}
{\mathbb E}{\mathrm D} &:={\mathbb E} \bigg\{ \frac{1}{\sqrt{N s_{1N}^2}} \sum_{k=1}^{N} {\mathrm D}(\Bdtheta_0;\mathbf{X}^{(k)}) \bigg\}\\
&= \frac{1}{\sqrt{N s_{1N}^2}} {\mathbb E} \sum_{k=1}^{N} \dfrac{\partial}{\partial\Bdtheta}{w_{ij} \log L_{ij}(X_{ti}^{(k)},X_{tj}^{(k)};{\Bdtheta})}\bigg|_{\Bdtheta=\Bdtheta_0}
=0 \mbox{ (no bias)}.
\end{align*}
Then by Theorem \ref{thm1}, (b) converges in distribution to $N(0,{\mathbb{V}}(\Bdtheta_0))$ where
\begin{align*}
{\mathbb{V}}(\Bdtheta_0) &={\mathbb{E}}[{\mathrm D} (\Bdtheta_0) {\mathrm D} (\Bdtheta_0)^T]\\
&= {\mathbb{E}}\big[({\mathrm D_1}+{\mathrm D_2}+{\mathrm D_3}+{\mathrm D_4})({\mathrm D_1}+{\mathrm D_2}+{\mathrm D_3}+{\mathrm D_4})^T \big]\\
&=\var[{\mathrm D_1}{\mathrm D_1}^T]+\var[{\mathrm D_2}{\mathrm D_2}^T]+\var[{\mathrm D_3}{\mathrm D_3}^T]+\var[{\mathrm D_4}{\mathrm D_4}^T], \\
\var[{\mathrm D_1}{\mathrm D_1}^T] &=\sigma(\Bdtheta_0;\mathbf{0})\\
& \quad +C_1\int\sigma\big(\Bdtheta_0; (0,\bm{h})\big)Q_2(\bm{h})d\bm{h}+C_1^2\int\sigma(\Bdtheta_0;\bm{h})Q_1(\bm{h})d\bm{h},
\end{align*}
and $\var[{\mathrm D_2}{\mathrm D_2}^T]$, $\var[{\mathrm D_3}{\mathrm D_3}^T]$ and $\var[{\mathrm D_4}{\mathrm D_4}^T]$ have similar forms with the variance of ${\mathrm D_1}{\mathrm D_1}^T$. Note that the event $\{x_i>u, x_j>u\}$ of $D_1$ is uncorrelated with the event $\{x_i>u, x_j \leq u\}$ of $D_2$, and $\cov(D_i,D_j)=0 \mbox{ for }i\neq j$.\\
Now suppose that $(X_{ti}^{(k)},X_{tj}^{(k)}),~k=1,\cdots,N$ are from $F_{b_k,d_k}$ not $H$. If $F \in D(G)$, there exists the exceedance level $(b_k,d_k)$ such that $F_{b_k,d_k}$ converges to $H$ as $(b_k,d_k) \rightarrow (x_0,y_0)$. The bivariate generalized pareto distribution $H$ preserves under the suitable change of exceedance levels (see \cite{rootzen:tajvidi:2006}).\\
The second-order condition (\ref{e:RV}) describes the difference between $F_{b_k,d_k}$ and $H$ with the remainder function $A(k)$, i.e., as $k \rightarrow \infty$, with the second order condition (ii)
\[ \limsup_{k\rightarrow \infty} | F_{b_k,d_k}(a_k x_i, c_k x_j)-H(x_i,x_j)|=O(A(k)). \]
Proposition \ref{prop} (ii) results from the condition (A6), and by the property of score function
\begin{align*}
{\mathbb E} & \bigg\{ \frac{1}{\sqrt{N s_{1N}^2}} \sum_{k=1}^{N} {\mathrm D}(\Bdtheta_0;\mathbf{X}^{(k)}) \bigg\}\\
& = \frac{1}{\sqrt{N s_{1N}^2}} \int \sum {\mathrm D}(\Bdtheta_0;\mathbf{X}^{(k)}) dF_{b_k,d_k} (a_k x_i,c_k x_j)\\
& = \sqrt{N s_{1N}^2} A(k) \cdot \frac{1}{N s_{1N}^2}\int \sum {\mathrm D}(\Bdtheta_0;\mathbf{X}^{(k)}) d\Psi(x_i,x_j)+ o(A(k))
\rightarrow \lambda \mu,
\end{align*}
where $\mu=\lim_{N\rightarrow \infty} \frac{1}{N s_{1N}^2} \int{\sum {\mathrm D}(\Bdtheta_0;\mathbf{X}^{(k)}) d\Psi(x_i,x_j)}$.\\
Then for some finite vector $\boldsymbol b$
\begin{eqnarray*}
{\big(N s_{1N}^2\big)}^{-1/2} {\mathbb E} \bigg\{\sum_{k=1}^{N} {\mathrm D}(\Bdtheta_0;\mathbf{X}^{(k)}) \bigg\}\rightarrow \boldsymbol b,
\end{eqnarray*}
and (b) converges in distribution to $N(\boldsymbol b,{\mathbb{V}}(\Bdtheta_0))$. Therefore the limit distribution of $\widehat \Bdtheta$, (\ref{e:normal2}) follows by Slutsky's Theorem. If the second-order condition (i) holds and $\sqrt{N s_{1N}^2}A(k_N)\rightarrow 0$, $\boldsymbol b=\mathbf{0}$ which implies no bias and then (\ref{e:normal1}) holds. 
\end{itemize}
\end{proof}

To prove consistency, we describe the theorem of \cite{amemiya:1985}.
\begin{theorem} \label{thm2}
(Amemiya, 1985) Assume the following:
\begin{itemize}
\item[(B1)] $\Theta$ is an open subset of Euclidean p-space (the true value $\theta_0$ is an interior point of $\Theta$),
\item[(B2)] The criterion function $S_N(\theta)$ is a measurable function for all $\theta_0 \in \Theta$, and $\nabla S_N$ exists and is continuous in an open neighborbood of $\theta_0$,
\item[(B3)] $\frac{1}{N}S_N(\theta)$ converges in probability uniformly to a non-stochastic function $S(\theta)$ in an open neighborhood of $\theta_0$, and $S(\theta)$ attains a strict local maximum at $\theta_0$.
\end{itemize}
Then there exists a sequence $\epsilon_N \rightarrow 0$ such that
\begin{eqnarray*}
P\{\exists \theta^* \mbox{ such that } |\theta^*-\theta_0|<\epsilon_N, \nabla S_N(\theta^*)=0\} \rightarrow 1, \mbox{ as }N\rightarrow\infty.
\end{eqnarray*}
\end{theorem}

\begin{theorem} (Consistency)
Let $\bm X_t^{(k)}=(X_{ti}^{(k)},X_{tj}^{(k)}),~k=1,\cdots,N$ be i.i.d. random variables with bivariate distribution $F$. Let $\hat\Bdtheta$ be the maximum pairwise composite log-likelihood estimator such that
\begin{eqnarray*}
\nabla S_N(\hat\Bdtheta):=\sum_{i<j}^K \sum_{t=1}^T w_{ij} \frac{\partial}{\partial \Bdtheta}{\log L(X_{ti},X_{tj};{\Bdtheta})}\bigg|_{\Bdtheta=\hat \Bdtheta}=0.
\end{eqnarray*}
If the second moment condition of composite score function is satisfied and conditions (A1), (B1) and (B2) hold, then there exists $\hat \Bdtheta$ such that $|\hat\Bdtheta-\Bdtheta_0|<\epsilon_N$ and $\nabla S_N(\hat\Bdtheta)=0$ for any sequence $\epsilon_N \rightarrow 0$, as $N\rightarrow\infty$ and $(b_k, d_k)=(b(k)_N,d(k)_N)\rightarrow(x_0,y_0)$.
\end{theorem}

\begin{proof}
Assumptions \textit{(B1)} and \textit{(B2)} in Theorem \ref{thm2} are satisfied by our criterion functions and assumptions. Jensen's inequality implies
\begin{equation} \label{e:jensen_ineq}
\int \log \bigg\{ \frac{f(x;\theta)}{f(x;\theta_0)} \bigg\} f(x;\theta_0)dx \leq \log \int f(x;\theta)dx=0.
\end{equation}
We rewrite it as
\begin{eqnarray*}
E_{\theta_0} \bigg[ \log \frac{f(x;\theta)}{f(x;\theta_0)} \bigg] \leq 0
\Leftrightarrow \theta_0= \arg\max_{\theta \in \Theta} E_{\theta_0} \bigg[ \log \frac{f(x;\theta)}{f(x;\theta_0)} \bigg].
\end{eqnarray*}
Here a sum of pairwise log-likelihoods can be considered. Let
\begin{eqnarray*}
S_N(\Bdtheta)=\sum_{i<j}^K \sum_{t=1}^T w_{ij} {\log L(X_{ti},X_{tj};{\Bdtheta})}.
\end{eqnarray*}
We know that by the law of large numbers,
\begin{align*}
\frac{1}{N} S_N(\Bdtheta_0)&=\frac{1}{N} \sum_{k=1}^N w_{ij}^{(k)} \log L(X_{ti}^{(k)},X_{tj}^{(k)};{\Bdtheta_0})
=\frac{1}{N} \sum_{k=1}^{N} w^{(k)}\log L(\bm X_{t}^{(k)};{\Bdtheta_0})\\
&{\longrightarrow}
E_{\Bdtheta_0} \big(w^{(1)} \log L(\bm X_{t}^{(1)};{\Bdtheta_0}) \big)=:S(\Bdtheta_0).
\end{align*}
By the moment condition of $\nabla S_N(\Bdtheta)$, we have that $E|\nabla S_N(\Bdtheta^*)|^2<C_0$ for some $C_0$. Using a Taylor's expansion,
\begin{align}
\Big|\Big|\Big(\frac{1}{N} S_N(\Bdtheta)-& S(\Bdtheta)\Big)- \Big(\frac{1}{N} S_N(\Bdtheta_0)-S(\Bdtheta_0)\Big)\Big|\Big|^2 \notag\\
&= \Big|\Big|\Big(\frac{1}{N} \nabla S_N(\Bdtheta^*)-\nabla S_N(\Bdtheta^{**})\Big)(\Bdtheta-\Bdtheta_0)\Big|\Big|^2 \notag\\
&\leq \bigg(\frac{1}{N}E|\nabla S_N(\Bdtheta^*)|^2+E|\nabla S_N(\Bdtheta^{**})|^2\bigg)||\Bdtheta-\Bdtheta_0||^2 \notag\\
&\leq \bigg(\frac{C_0}{N}+C_0\bigg) ||\Bdtheta-\Bdtheta_0||^2 \notag\\
& \longrightarrow C_0 ||\Bdtheta-\Bdtheta_0||^2 \label{e:thm3_diff}
\end{align}
for some $\Bdtheta^*$ and $\Bdtheta^{**}$ between $\Bdtheta_0$ and $\Bdtheta$. By the moment condition of $\nabla S_N(\Bdtheta)$, the right hand side of (\ref{e:thm3_diff}) converges to 0 uniformly over a sequence of $||\Bdtheta-\Bdtheta_0||<\epsilon_N$ as $\epsilon_N \rightarrow 0$. Also we have that $\frac{1}{N} S_N(\Bdtheta_0)-S(\Bdtheta_0)\stackrel {p}{\longrightarrow} 0$ by the law of large numbers and $\frac{1}{N} S_N(\Bdtheta)$ converges in probability uniformly to $S(\Bdtheta)$ on a neighborhood of $\Bdtheta_0$.

Now we claim that $S(\Bdtheta)$ attains a local maximum at $\Bdtheta=\Bdtheta_0$. The previous result (\ref{e:jensen_ineq}) implies that
\begin{eqnarray*}
E_{\Bdtheta_0} \bigg[ \log \frac{\prod_{k} L(\bm X_{t}^{(k)};{\Bdtheta})}{\prod_{k} L(\bm X_{t}^{(k)};{\Bdtheta_0})} \bigg] \leq
\log E_{\Bdtheta_0} \bigg[ \frac{\prod_{k} L(\bm X_{t}^{(k)};{\Bdtheta})}{\prod_{k} L(\bm X_{t}^{(k)};{\Bdtheta_0})} \bigg]=0
\end{eqnarray*} and for any $\Bdtheta$,
\begin{eqnarray*}
E_{\Bdtheta_0} \big(\log \prod_{k} L(\bm X_{t}^{(k)};{\Bdtheta_0}) \big) \geq
E_{\Bdtheta_0} \big(\log \prod_{k} L(\bm X_{t}^{(k)};{\Bdtheta}) \big).
\end{eqnarray*}
where the equality holds with (A1), the identifiability assumption of parameter.\\
$E_{\Bdtheta_0} \big[ \log \frac{\prod_{k} L(\bm X_{t}^{(k)};{\Bdtheta})}{\prod_{k} L(\bm X_{t}^{(k)};{\Bdtheta_0})} \big] \leq 0$ holds for any distribution of $\bm X_t^{(k)}$ with finite second moments of score function, and the maximum of $E_{\Bdtheta_0}[ \log \prod_{k} L(\bm X_{t}^{(k)};{\Bdtheta})]$ over $\Bdtheta$ is attained at $\Bdtheta=\Bdtheta_0$. Thus we prove the \textit{(B3)} of the Theorem \ref{thm2}.
\end{proof}

\subsection{Simulation}
We conduct some simulation studies to illustrate the asymptotic behavior of the estimators described in Section \ref{sec:3:1}. The simulation is examined for the daily max-stable process with unit Fr\'{e}chet margins with $T=1000$ days during 10 years, i.e., $M=100$ in equation (\ref{e:DAAM}). We consider the Gaussian extreme value processes with two different spatial dependence structures of the covariance matrix:
\begin{itemize}
\item[] \[ \Sigma=\left( \begin{array}{cc}
\alpha & \beta \\
\beta & \gamma \end{array} \right)\]
\item[(i)] the Gaussian extreme value process with $\Sigma_1$ ($\alpha=2$, $\beta=0$ and $\gamma=3$);
\item[(ii)] the Gaussian extreme value process with $\Sigma_2$ ($\alpha=2$, $\beta=1.5$ and $\gamma=3$).
\end{itemize}
We generate $n=20$ stations from the uniform density function $f(\cdot)$ over $R_0=(-1/2,1/2]^d$ and determine the growth rate $\lambda_n=\sqrt{n}$ in case of $d=2$ to satisfy the relation $n\sim C\lambda_n^d$ in the spatial structure and stochastic sampling design of sites. To adjust the threshold approach based on the pairwise composite likelihood, we consider a weight function such that for some constant $\delta_0$,
\[
w(h)=
\begin{cases}
1 & \text{if } h \leq \delta_0\\
0 & \text{if } h > \delta_0.
\end{cases}
\]
where $h$ is a distance between two stations. Here $\delta_0$ is selected by $\sqrt{2n}/2$, the half diagonal of sampling region, which satisfies the condition (A$^\prime$6) on growth rate of weight function for the asymptotic result.

To illustrate the asymptotic performance of estimates for dependence parameter $\Bdtheta=(\alpha,\beta,\gamma)$, the averages of the estimators are compared to the asymptotic mean of $\hat{\Bdtheta}$. In each model, the estimation of dependence parameters is based on 500 replications, and the classical Monte Carlo integration is used to implement the theoretical bias and variance of the estimators as the number of exceedances $N$ increases.

Theoretical bias and average bias of estimators $\hat{\Bdtheta}$ for Smith model (i) are plotted in Figure \ref{f:temp1}. As the number of exceedances increases, bias of estimators (gray curve) tends to decrease towards the theoretical bias (solid curve) though each estimator shows the different slope on the decay. The bias of $\hat\alpha$ goes on with the pattern of decay of theoretical one, while bias of $\hat\beta$ and $\hat\gamma$ decreases as theoretical bias goes up to the line of zero bias.

\begin{figure}
\includegraphics[width=0.85\textwidth]{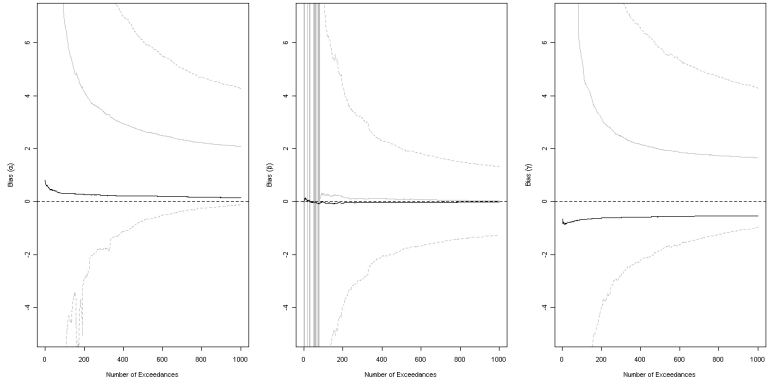}
\caption[Graphical summary of asymptotic behavior for Smith model (i)]{Graphical summary of asymptotic behaviors of $\hat\alpha$, $\hat\beta$ and $\hat\gamma$ for Smith model (i) from left to right. Gray curve is the average bias of estimators, gray dashed curves are the boundary of 95\% confidence interval, and black solid curve is the theoretical bias.}\label{f:temp1}
\end{figure}

\begin{figure}
\includegraphics[width=0.85\textwidth]{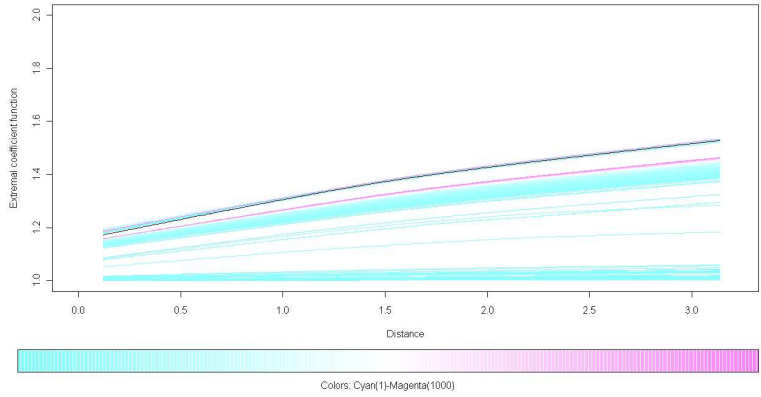}
\caption[Extremal coefficient functions for the Smith model (i)]{Extremal coefficient functions for the Smith model (i). Upper thin color layer is based on theoretical mean of estimates and lower thick color layer is based on average estimates. In a layer, each line represent a extremal coefficient curve at each $N$ and the line changes the color from cyan ($N=1$) to magenta ($N=1000$). Black solid line is the true extremal coefficient curve.}\label{f:extCoef1}
\end{figure}

This irregular pattern of each dependence parameter estimation might be caused by the interaction between parameters in estimating them as components of covariance matrix. Now we plot the extremal coefficient curves with the parameter estimators and compare them with those estimated directly. One can expect the problem to be reduced when working with the extremal coefficient. 

\begin{figure}
\includegraphics[width=0.85\textwidth]{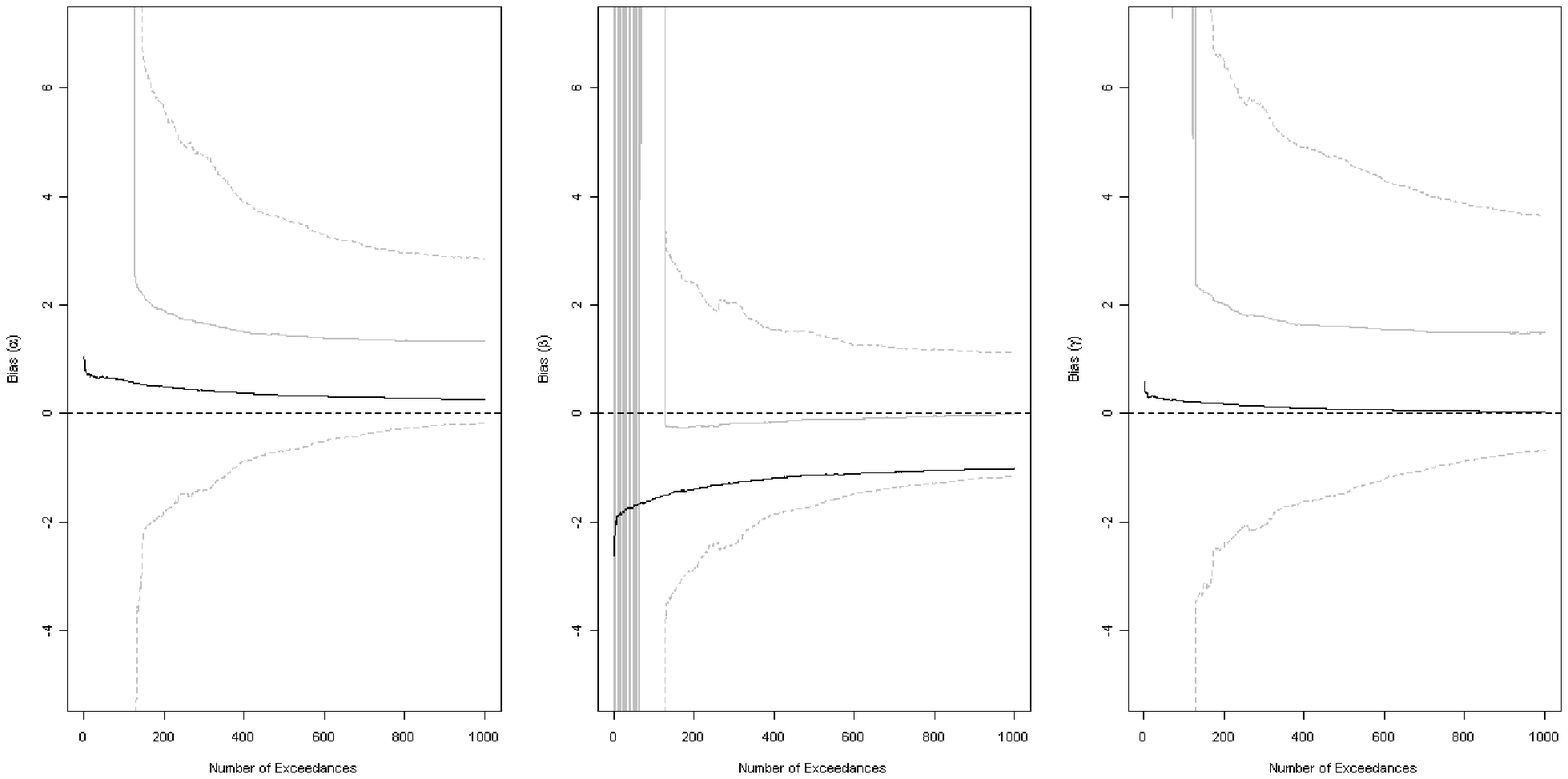}
\caption[Graphical summary of asymptotic behavior for Smith model (ii)]{Graphical summary of asymptotic behaviors of $\hat\alpha$, $\hat\beta$ and $\hat\gamma$ for Smith model (ii) from left to right. Gray curve is the average bias of estimators, gray dashed curves are the boundary of 95\% confidence interval, and black solid curve is the theoretical bias.}\label{f:temp2}
\end{figure}

\begin{figure}
\includegraphics[width=0.85\textwidth]{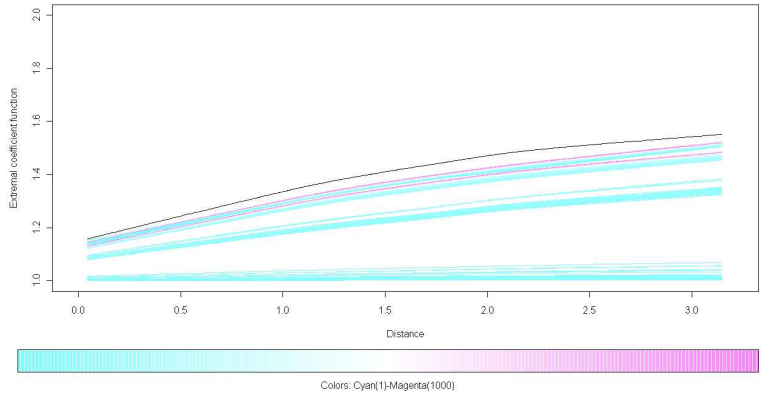}
\caption[Extremal coefficient functions for the Smith model (ii)]{Extremal coefficient functions for the Smith model (ii). Upper thin color layer is based on theoretical mean of estimates and lower thick color layer is based on average estimates. In a layer, each line represent a extremal coefficient curve at each $N$ and the line changes the color from cyan ($N=1$) to magenta ($N=1000$). Black solid line is the true extremal coefficient curve.}\label{f:extCoef2}
\end{figure}

Figure \ref{f:extCoef1} shows estimated extremal coefficient functions by $\hat\Bdtheta$. As the number of exceedances increases, the color changes from cyan to magenta. Extremal coefficient by the asymptotic bias overlapped almost with the true coefficient function (black solid curve). As the number of exceedances increases, The extremal coefficient curve measured by dependence estimators approximates the theoretical extremal coefficient curve. However, there still exists a gap between the theoretical extremal coefficient and estimated one and the gap gets broader as the distance between two locations is larger.

Theoretical bias and average bias of estimators $\hat{\Bdtheta}$ for Smith model (ii) are shown in Figure \ref{f:temp2}. As the number of exceedances increases, bias of estimates tends to go towards the pattern of theoretical bias. There is some gaps between theoretical bias and estimated bias though the estimation of dependence parameter is much more stable comparing with that in model (i).

Figure \ref{f:extCoef2} shows estimated extremal coefficient functions by $\hat\Bdtheta$. As the number of exceedances increases, The extremal coefficient curve measured by dependence estimators approximates the theoretical extremal coefficient curve. Unlike the gap in Figure \ref{f:temp2}, the estimated extremal coefficient is catching up with the theoretical one along by a little gap. However, the quality of asymptotic approximation seems dependent on the degree of correlation $\beta$ since Figure \ref{f:extCoef2} shows the poor approximation to the true extremal coefficient curve comparing with Figure \ref{f:extCoef1}.

Suggestion on the choice of the threshold point is discussed further now. For the simplicity, the threshold can be selected as the value of the 95th percentile of distribution function in practice. However finding an optimal threshold is another important issue and we suggest an optimal threshold minimizing the mean squared error, which incorporates both the bias of the estimator and its variance based on the asymptotic normality in Section \ref{sec:3:1}.

\begin{figure}
\includegraphics[width=0.85\textwidth]{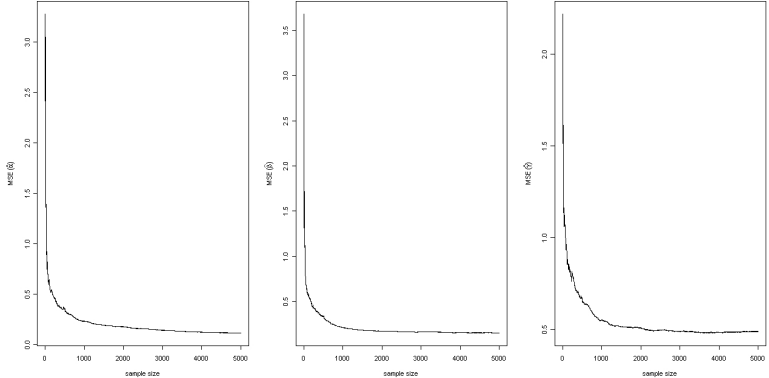}
\caption[Mean squared error of $\hat\alpha$, $\hat\beta$ and $\hat\gamma$ for Smith (i) from left to right]{Mean squared error of $\hat\alpha$, $\hat\beta$ and $\hat\gamma$ for Smith (i) from left to right}\label{f:opt1}
\end{figure}

\begin{figure}
\includegraphics[width=0.85\textwidth]{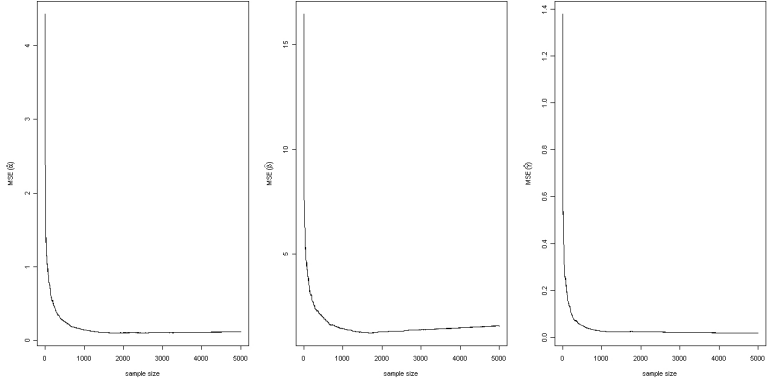}
\caption[Mean squared error of $\hat\alpha$, $\hat\beta$ and $\hat\gamma$ for Smith (ii) from left to right]{Mean squared error of $\hat\alpha$, $\hat\beta$ and $\hat\gamma$ for Smith (ii) from left to right}\label{f:opt2}
\end{figure}

Figure \ref{f:opt1} shows the mean squared error for each estimator in Smith model (i). The mean squared errors of $\hat\alpha$, $\hat\beta$, and $\hat\gamma$ are decreasing rapidly against $N$ and show the stability between $N=3500$ and $N=5000$. As shown in Figure \ref{f:extCoef1}, the theoretical extremal coefficient has a nice approximation to the true coefficient function, and the increases of squared bias seem to be less effective than variance decreases on the selection of threshold to minimize the MSE.

The mean squared error for each estimator of Smith model (ii) is shown in Figure \ref{f:opt2}. The mean squared errors of $\hat\alpha$, $\hat\beta$, and $\hat\gamma$ are decreasing rapidly as $N$ increases to 1000, and have the minimum between $N=1500$ and $N=2000$. In Figure \ref{f:extCoef2}, the theoretical extremal coefficient shows the poor approximation to the true coefficient function. Thus calculation of MSE is affected by the increase of bias as number of exceedances becomes greater than 1500. The threshold point is suggested as the value between 90th and 95th percentile. 

\section{Discussion}
\label{discn}
The threshold approach takes advantage of avoiding the loss of information which is caused when we are concerned with only maxima of data. Our method is expected to become one promising tool to characterize the dependence structure in spatial extremes. we have suggested the modeling of the bivariate exceedances over threshold and it leads to a simplified dependence structure for max-stable processes. An important motivation of this methodology is the possibility of threshold approach to construct approximation of the joint distribution, by assuming an asymptotic distribution of exceedances over a given threshold. We have derived our simulation results under two Smith models to examine the asymptotic property of estimates.

Moreover, we have also investigated an optimal threshold to minimize the mean squared error based on the asymptotic behavior of the estimator for dependence parameter. The choice of optimal threshold would be an open topic itself for further research. It provides very valuable information in the field of environmental statistics. When we are interested in flooding, for example, which may be considered as extreme events, choosing the adequate threshold to avoid the risk of flooding might be useful for quantifying the spatial extremal dependence.

\section*{Appendix}

\textbf{A. Example of Proposition 1}

Suppose that $(X,Y)$ are i.i.d. from a bivariate normal distribution $F$ with mean 0, variance 1 and correlation coefficient $\rho$. First we would like to prove that bivariate normal distribution satisfies (18) in the paper. We consider $G$ in (14) as a bivariate extreme value distribution with Gumbel margins, and suppose the limiting form of bivariate normal $G(x,y)=\exp\{-e^{-x}-e^{-y}\}$ in the case of the independence. A max-stable process with unit Fr\'{e}chet margins will be fitted and the transformations $X'=\log{X}$ and $Y'=\log{Y}$ can be made from unit Fr\'{e}chet to Gumbel.

Mills ratio for a normal density implies that
\begin{align*}
\frac{1-\Phi(x)}{\phi(x)} & \sim \bigg\{\frac{1}{x}-\frac{1}{x^3}+\frac{1\cdot 3}{x^5}-\frac{1\cdot 3\cdot 5}{x^7}+\cdots\bigg\},\\
\frac{P(X>x,Y>y)}{\phi(x,y)} & \sim \frac{(1-\rho^2)^2}{(x-\rho y)(y-\rho x)} \times\\
& \qquad \bigg\{1-(1-\rho^2)\bigg(\frac{1}{(x-\rho y)^2}-\frac{\rho}{(x-\rho y)(y-\rho x)}+\frac{1}{(y-\rho x)^2}\bigg)+\cdots\bigg\}
\end{align*} (see \cite{ruben:1964} for the bivariate normal density). From the fact that
\beq
1-F(x,y)=1-\Phi(x)+1-\Phi(y)-P(X>x,Y>y),
\eeq we could set the lower bound and upper bound for $\frac{1-F(x,y)}{\phi(x,y)}$ such that
\begin{align*}
\bigg(\frac{1-F(x,y)}{\phi(x,y)}\bigg)^L & \leq \frac{1-F(x,y)}{\phi(x,y)} \leq \bigg(\frac{1-F(x,y)}{\phi(x,y)}\bigg)^U,\\
\mbox{where } \bigg(\frac{1-F(x,y)}{\phi(x,y)}\bigg)^L &=\frac{1}{x}+\frac{1}{y}-\frac{1}{x^3}-\frac{1}{y^3}-\frac{(1-\rho^2)^2}{(x-\rho y)(y-\rho x)},\\
\bigg(\frac{1-F(x,y)}{\phi(x,y)}\bigg)^U &=\frac{1}{x}+\frac{1}{y}-\frac{(1-\rho^2)^2}{(x-\rho y)(y-\rho x)}+\frac{(1-\rho^2)^3}{(x-\rho y)(y-\rho x)} \times\\
& \qquad \bigg(\frac{1}{(x-\rho y)^2}-\frac{\rho}{(x-\rho y)(y-\rho x)}+\frac{1}{(y-\rho x)^2}\bigg).
\end{align*}

From the well-known results of extreme value theory, define $b_t$ by $1-\Phi(b_t)=\frac{1}{t}$ and $a_t=1/b_t$. Or we might set normalized constants
\begin{align*}
a_t &=\frac{1}{\sqrt{2\log t}}\\
b_t &=\sqrt{2\log t}-\frac{\frac{1}{2}(\log \log t+\log 4\pi)}{\sqrt{2\log t}}.
\end{align*}

Conditional distribution of exceedances over threshold is written as
\[
F_{b_t, d_t}(a_t x, c_t y) = 1-\frac{t \big\{1-F(a_t x+b_t, c_t y+d_t) \big\}}{t \big\{1-F(b_t, d_t) \big\}}
\]
and we now concentrate on $\frac{1-F(a_t x+b_t, c_t y+d_t)}{1-F(b_t, d_t)}$,
\begin{align*}
&\frac{1-F(a_t x+b_t, c_t y+d_t)}{1-F(b_t, d_t)}=\frac{1-F(a_t x+b_t, c_t y+d_t)}{\phi(a_t x+b_t, c_t y+d_t)}\cdot \frac{\phi(b_t,d_t)}{1-F(b_t, d_t)}\cdot \frac{\phi(a_t x+b_t, c_t y+d_t)}{\phi(b_t,d_t)}\\
&\quad \geq \bigg\{\frac{b_t}{x+b_t^2}+\frac{d_t}{y+d_t^2}-\frac{b_t^3}{(x+b_t^2)^3}-\frac{b_t^3}{(y+b_t^2)^3}
-\frac{(1-\rho^2)^2 b_t^2}{\big(x-\rho y+b_t^2(1-\rho)\big) \big(y-\rho x+b_t^2(1-\rho)\big)} \bigg\}\\
&\qquad \times \bigg\{\frac{b_t^4}{2 b_t^3-(1+\rho)^2 b_t^2+\frac{(1+\rho)^3(2-\rho)}{1-\rho}}\bigg\}
\frac{\phi(a_t x+b_t, c_t y+d_t)}{\phi(b_t,d_t)}\\
& \quad \sim \bigg\{ \frac{2 b_t^2-(1+\rho)^2 b_t-2}{b_t^3} \bigg\} \bigg\{\frac{b_t^4}{2 b_t^3-(1+\rho)^2 b_t^2+\frac{(1+\rho)^3(2-\rho)}{1-\rho}}\bigg\} \frac{\phi(x/b_t+b_t, y/b_t+b_t)}{\phi(b_t,b_t)}
\end{align*}
and also,
\begin{align*}
&\frac{1-F(a_t x+b_t, c_t y+d_t)}{1-F(b_t, d_t)}\\
&\quad \leq \bigg[\frac{b_t}{x+b_t^2}+\frac{d_t}{y+d_t^2}-\frac{(1-\rho^2)^2 b_t^2}{\big(x-\rho y+b_t^2(1-\rho)\big) \big(y-\rho x+b_t^2(1-\rho)\big)} \bigg\{ 1-(1-\rho^2) \times\\
&\qquad \bigg(\frac{b_t^2}{\big(x-\rho y+b_t^2(1-\rho)\big)^2}+\frac{b_t^2}{\big(y-\rho x+b_t^2(1-\rho)\big)^2}\\
&\qquad -\frac{\rho b_t^2}{\big(x-\rho y+b_t^2(1-\rho)\big) \big(y-\rho x+b_t^2(1-\rho)\big)}\bigg)\bigg\}\bigg]\bigg(\frac{b_t^3}{2 b_t^2-(1+\rho)^2 b_t-2}\bigg)\\
&\qquad \times \frac{\phi(a_t x+b_t, c_t y+d_t)}{\phi(b_t,d_t)}\\
&\quad \sim \bigg\{\frac{2 b_t^3-(1+\rho)^2 b_t^2+\frac{(1+\rho)^3(2-\rho)}{1-\rho}}{b_t^4}\bigg\} \bigg\{\frac{b_t^3}{2 b_t^2-(1+\rho)^2 b_t-2}\bigg\} \frac{\phi(x/b_t+b_t, y/b_t+b_t)}{\phi(b_t,b_t)}.
\end{align*}
Thus
\begin{align*}
&F_{b_t, d_t}(a_t x, c_t y)-H(x,y) = -\frac{1-F(a_t x+b_t, c_t y+d_t)}{1-F(b_t, d_t)}+(e^{-x}+e^{-y})\\
& \sim \bigg\{-\frac{2b_t^3-(1+\rho)^2 b_t^2-2b_t}{2b_t^3-(1+\rho)^2 b_t^2+\frac{(1+\rho)^3(2-\rho)}{1-\rho}}+1 \bigg\} 
\bigg\{\frac{\phi(x/b_t+b_t, y/b_t+b_t)}{\phi(b_t,b_t)}+e^{-x}+e^{-y}\bigg\}
\end{align*}
and
\[-\frac{2b_t^3-(1+\rho)^2 b_t^2-2b_t}{2b_t^3-(1+\rho)^2 b_t^2+\frac{(1+\rho)^3(2-\rho)}{1-\rho}}+1
=\frac{2b_t+\frac{(1+\rho)^3(2-\rho)}{1-\rho}}{2b_t^3-(1+\rho)^2 b_t^2+\frac{(1+\rho)^3(2-\rho)}{1-\rho}}.\]
We obtain the formation of (18),
\begin{align*}
\lim_{t\rightarrow \infty} \frac{F_{b_t, d_t}(a_t x, c_t y)-H(x,y)}{A(t)} = \Psi(x,y)
\end{align*}
where $A(t)=\frac{1}{b_t^2}=\frac{1}{2\log t}$ and $\Psi(x,y)=\exp\big\{-\frac{x+y}{1+\rho}\big\}+e^{-x}+e^{-y}$.

Next,
\begin{align} \label{e:exam_f}
\nonumber f_{b_t,d_t}&(a_t x,c_t y)= \frac{a_t c_t}{1-F(b_t,d_t)} \cdot \frac{1}{2\pi \sqrt{1-\rho^2}} \times\\
\nonumber &\qquad \exp \bigg\{-\frac{(a_t x+b_t)^2+(c_t y+d_t)^2-2\rho (a_t x+b_t)(c_t y+d_t)}{2(1-\rho^2)}\bigg\}\\
& =a_t c_t \frac{\phi(b_t,d_t)}{1-F(b_t,d_t)} \cdot \frac{\phi(a_t x+b_t,c_t y+d_t)}{\phi(b_t,d_t)}
\doteq a_t c_t \frac{\phi(b_t,d_t)}{1-F(b_t,d_t)} \cdot V_t(x,y)
\end{align}
where $\phi(x,y)$ is a bivariate normal density with correlation $\rho$.

$\frac{\phi(x,y)}{1-F(x,y)}$ as a factor of $f_{b_t,d_t}(a_t x,c_t y)$ in the equation (\ref{e:exam_f}) has the lower and upper bounds that
\begin{align*}
\bigg(\frac{\phi(x,y)}{1-F(x,y)}\bigg)^L & \leq \frac{\phi(x,y)}{1-F(x,y)} \leq \bigg(\frac{\phi(x,y)}{1-F(x,y)}\bigg)^U,
\end{align*}
where
\begin{align*}
\bigg(\frac{\phi(x,y)}{1-F(x,y)}\bigg)^L &=\bigg\{\frac{1}{x}+\frac{1}{y}-\frac{(1-\rho^2)^2}{(x-\rho y)(y-\rho x)}+\frac{(1-\rho^2)^3}{(x-\rho y)(y-\rho x)} \times\\
& \qquad \bigg(\frac{1}{(x-\rho y)^2}-\frac{\rho}{(x-\rho y)(y-\rho x)}+\frac{1}{(y-\rho x)^2}\bigg)\bigg\}^{-1},\\
\bigg(\frac{\phi(x,y)}{1-F(x,y)}\bigg)^U &=\bigg\{\frac{1}{x}+\frac{1}{y}-\frac{1}{x^3}-\frac{1}{y^3}-\frac{(1-\rho^2)^2}{(x-\rho y)(y-\rho x)}\bigg\}^{-1}.
\end{align*}
Since $f_{b_t,d_t}(a_t x,c_t y)=a_t c_t \frac{\phi(b_t,d_t)}{1-F(b_t,d_t)} \cdot V_t(x,y)$, using above normalized constants and assuming $b_t=d_t$
\begin{align*}
a_t c_t \bigg(\frac{\phi(b_t,d_t)}{1-F(b_t,d_t)}\bigg)^L
&=\frac{b_t^2}{2 b_t^3-(1+\rho)^2 b_t^2+\frac{(1+\rho)^3(2-\rho)}{1-\rho}}\\
a_t c_t \bigg(\frac{\phi(b_t,d_t)}{1-F(b_t,d_t)}\bigg)^U
&= \frac{b_t}{2 b_t^2-(1+\rho)^2 b_t-2}\\
V_t(x,y) & = \frac{\phi(x/b_t+b_t,y/b_t+b_t)}{\phi(b_t,b_t)}=\exp \bigg\{-\frac{x^2+y^2-2\rho xy}{2(1-\rho^2)b_t^2}-\frac{x+y}{1+\rho}\bigg\}.
\end{align*}
Thus we could get the following form of bounds
\begin{align*}
f_{b_t,d_t}^L (a_t x,c_t y) & = \frac{b_t^2}{2 b_t^3-(1+\rho)^2 b_t^2+\frac{(1+\rho)^3(2-\rho)}{1-\rho}} \frac{\phi(x/b_t+b_t,y/b_t+b_t)}{\phi(b_t,b_t)}\\
f_{b_t,d_t}^U (a_t x,c_t y) &= \frac{b_t}{2 b_t^2-(1+\rho)^2 b_t-2} \frac{\phi(x/b_t+b_t,y/b_t+b_t)}{\phi(b_t,b_t)}.
\end{align*}
Meanwhile
\begin{align*}
h(x,y) = \frac{\partial^2 H(x,y)}{\partial x \partial y}=-\frac{1}{\log G(0,0)}\cdot \frac{\partial^2}{\partial x \partial y} \log G(x,y)=0.
\end{align*}
Therefore
\begin{align*}
f_{b_t,d_t}(a_t x &,c_t y) -h(x,y) \geq \{f_{b_t,d_t}(a_t x,c_t y)-h(x,y)\}^L \\
& = \frac{b_t^2}{2 b_t^3-(1+\rho)^2 b_t^2+\frac{(1+\rho)^3(2-\rho)}{1-\rho}} \frac{\phi(x/b_t+b_t,y/b_t+b_t)}{\phi(b_t,b_t)},\\
f_{b_t,d_t}(a_t x &,c_t y) -h(x,y) \leq \{f_{b_t,d_t}(a_t x,c_t y)-h(x,y)\}^U \\
&= \frac{b_t}{2 b_t^2-(1+\rho)^2 b_t-2} \frac{\phi(x/b_t+b_t,y/b_t+b_t)}{\phi(b_t,b_t)}.
\end{align*}
Define $A(t)=\frac{1}{2\log t}$ ($A(t)\rightarrow 0$ as $t\rightarrow \infty$) and $\psi(x,y)=-\frac{\rho}{2(1-\rho^2)}$ to satisfy the condition (18). Then we could show that
\begin{align*}
\frac{f_{b_t,d_t}(a_t x,c_t y)-h(x,y)}{A(t)}& -\psi(x,y) \geq \frac{f_{b_t,d_t}(a_t x,c_t y)^L-h(x,y)}{1/(2 \log t)}-\psi(x,y) \\
& \sim \exp\bigg(-\frac{x+y}{1+\rho}\bigg) \bigg\{\frac{b_t}{2} \exp\bigg(-a_t^2 \frac{x^2+y^2-2\rho xy}{2(1-\rho^2)}\bigg)-\frac{1}{(1+\rho)^2}\bigg\},\\
\frac{f_{b_t,d_t}(a_t x,c_t y)-h(x,y)}{A(t)}& -\psi(x,y) \leq \frac{f_{b_t,d_t}(a_t x,c_t y)^U-h(x,y)}{1/(2 \log t)}-\psi(x,y) \\
& \sim \exp\bigg(-\frac{x+y}{1+\rho}\bigg) \bigg\{\frac{b_t}{2} \exp\bigg(-a_t^2 \frac{x^2+y^2-2\rho xy}{2(1-\rho^2)}\bigg)-\frac{1}{(1+\rho)^2}\bigg\}.
\end{align*}
This limit for bounds of $\frac{f_{b_t,d_t}-h}{A(t)}-\psi(x,y)$ will be used to prove that the product of a function $g_t(x,y)$ and $\frac{f_{b_t,d_t}-h}{A(t)}-\psi(x,y)$ is bounded by an integrable function as shown in (19), Proposition 1. Suppose that $g_t(x,y)=\frac{\partial}{\partial \theta}\log f_{DA}(x,y;\theta)$ where $f_{DA}=\frac{\partial^2 F_{DA}(x,y)}{\partial x \partial y}$. Any max-stable process can be fitted for modeling annual maxima of data and we can obtain the score function by our threshold method with the composite likelihood approach. We arbitrarily choose the Brown-Resnick process with Gumbel margins to obtain the joint bivariate distribution of annual data, $F_{AM}$, and a joint bivariate distribution of daily data, $F_{DA}(x,y)$, is determined by the relation (12).
\[
F_{AM}(x,y;\theta) =\exp \{B(x,y;\theta)\},
\]
where $B(x,y;\theta) =\Big\{-\frac{1}{x}\Phi\Big(\frac{\sqrt{\gamma(h;\theta)}}{2}+\frac{1}{\sqrt{\gamma(h;\theta)}}\log\frac{y}{x}\Big)
-\frac{1}{y}\Phi\Big(\frac{\sqrt{\gamma(h;\theta)}}{2}+\frac{1}{\sqrt{\gamma(h;\theta)}}\log\frac{x}{y}\Big)\Big\}$ and
\[
\log f_{DA}(x,y;\theta) = \frac{1}{M}B(x,y;\theta)+\log J(x,y;\theta),
\]
where $J(x,y;\theta) =\frac{1}{M} \frac{\partial^2 B(x,y;\theta)}{\partial x \partial y}
+\frac{1}{M^2} \frac{\partial B(x,y;\theta)}{\partial x}\cdot\frac{\partial B(x,y;\theta)}{\partial y}$.
Therefore,
\be \label{e:g_t}
g_t(x,y) =\frac{1}{M} \frac{\partial B(x,y;\theta)}{\partial \theta}+J(x,y;\theta)^{-1} \bigg(\frac{\partial J(x,y;\theta)}{\partial \theta}\bigg)
\ee
where $\frac{\partial J(\theta)}{\partial \theta} =\frac{1}{M} \frac{\partial}{\partial \theta}\Big(\frac{\partial^2 B(x,y;\theta)}{\partial x \partial y}\Big) +\frac{1}{M^2} \frac{\partial}{\partial \theta}\Big(\frac{\partial B(x,y;\theta)}{\partial x}\Big) \cdot \frac{\partial B(x,y;\theta)}{\partial y}
+\frac{1}{M^2} \frac{\partial B(x,y;\theta)}{\partial x}\cdot\frac{\partial}{\partial \theta}\Big(\frac{\partial B(x,y;\theta)}{\partial y}\Big)$.
With some calculations, the derivatives of $J(x,y;\theta)$ and $B(x,y;\theta)$, shortly $J$ and $B$, can be obtained as in Appendix B and the boundness of the product is of interest:
\begin{align} \label{e:exam_bound_again}
\nonumber \bigg| g_t(x,y) & \bigg\{\frac{f_{b_t,d_t}(a_t x,c_t y)-h(x,y)}{A(t)}-\psi(x,y) \bigg\} \bigg| \\
& \leq \bigg| g_t(x,y)
\exp\bigg(-\frac{x+y}{1+\rho}\bigg) \bigg\{\frac{b_t}{2} \exp\bigg(-a_t^2 \frac{x^2+y^2-2\rho xy}{2(1-\rho^2)}\bigg)-\frac{1}{(1+\rho)^2}\bigg\} \bigg|.
\end{align}

\begin{enumerate}
\item[\textbf{Case (i)}] $x=y$:
\begin{align*}
\frac{\partial B}{\partial \theta} &= \bigg(\frac{\partial \gamma}{\partial \theta}\bigg)
\bigg\{-e^{-x} \bigg(\frac{1}{2 \sqrt\gamma}\bigg) \phi \bigg(\frac{\sqrt\gamma}{2}\bigg) \bigg\},\\
J(\theta) &=\frac{\sqrt\gamma}{M} e^{-x} \phi \bigg(\frac{\sqrt\gamma}{2}\bigg)+\frac{1}{M^2} e^{-2x} \bigg\{\Phi^2 \bigg(\frac{\sqrt\gamma}{2}\bigg)-\frac{2}{\gamma} \phi^2 \bigg(\frac{\sqrt\gamma}{2}\bigg) \bigg\},\\ \bigg\}.
\frac{\partial J}{\partial \theta} &= \bigg(\frac{\partial \gamma}{\partial \theta}\bigg)
\bigg\{-\frac{1}{M} \bigg(\frac{1}{8\sqrt\gamma}+\frac{1}{2 \sqrt{\gamma^3}}\bigg) e^{-x} \phi \bigg(\frac{\sqrt\gamma}{2}\bigg)
+\frac{1}{M^2} \bigg(\frac{1}{2 \sqrt\gamma}\bigg) e^{-2x} \phi \bigg(\frac{\sqrt\gamma}{2}\bigg) \Phi \bigg(\frac{\sqrt\gamma}{2}\bigg) \bigg\}.
\end{align*}
Then
\begin{align*}
g_t(x,y) & \leq \bigg(\frac{\partial \gamma}{\partial \theta}\bigg)
\bigg\{-\frac{1}{M} \bigg(\frac{1}{2 \sqrt\gamma}\bigg) \bigg( 1-
\frac{\Phi\big(\frac{\sqrt\gamma}{2}\big)}{ \sqrt\gamma \phi\big(\frac{\sqrt\gamma}{2}\big)
+\frac{e^{-x}}{M} \big\{\Phi^2\big(\frac{\sqrt\gamma}{2}\big)-\frac{2}{\gamma} \phi^2 \big(\frac{\sqrt\gamma}{2}\big) \big\}}\bigg) \phi\bigg(\frac{\sqrt\gamma}{2}\bigg) \bigg\} e^{-x}
\end{align*}
and therefore, for some constants $C_i$
\begin{align*}
\bigg| g_t(x,y) \bigg\{ &\frac{f_{b_t,d_t}(a_t x,c_t y)-h(x,y)}{A(t)} -\psi(x,y) \bigg\} \bigg|\\
&\qquad \leq C_1 \bigg(\frac{\partial \gamma}{\partial \theta}\bigg) e^{-x} e^{-\frac{2x}{1+\rho}}
\bigg\{ b_t \exp\bigg(-\frac{x^2}{(1+\rho)b_t^2}\bigg)-\frac{2}{(1+\rho)^2}\bigg\}\\
&\qquad \leq C_2 \phi \bigg(\frac{\sqrt{2}x}{\sqrt{1+\rho}b_t}+\frac{b_t(3+\rho)}{\sqrt{2(1+\rho)}}\bigg),
\end{align*}
which implies that (\ref{e:exam_bound_again}) is bounded by an integrable function.

\item[\textbf{Case (ii)}] $y=x+k$ and $x \rightarrow \infty$:\\
Let
\begin{align*}
\frac{\sqrt\gamma}{2}+\frac{1}{\sqrt\gamma}(y-x) &=\frac{\sqrt\gamma}{2}+\frac{k}{\sqrt\gamma}=a,\\
\frac{\sqrt\gamma}{2}+\frac{1}{\sqrt\gamma}(x-y) &=\frac{\sqrt\gamma}{2}-\frac{k}{\sqrt\gamma}=b,\\
\frac{1}{4\sqrt\gamma}-\frac{1}{2\sqrt{\gamma^3}}(x-y) &=\frac{1}{2\gamma}\logyx=\frac{1}{2\gamma}a,\\
\frac{1}{4\sqrt\gamma}-\frac{1}{2\sqrt{\gamma^3}}(y-x) &=\frac{1}{2\gamma}\logxy=\frac{1}{2\gamma}b.
\end{align*}
\begin{align*}
\frac{\partial B}{\partial \theta} &=\bigg(\frac{\partial \gamma}{\partial \theta}\bigg)
\bigg\{-e^{-x}\bigg(\frac{1}{2\gamma}\bigg)\big(b\phi(a)+e^{-k}a\phi(b)\big)\bigg\},
\end{align*}
\begin{align*}
J(\theta) &=\frac{1}{M} e^{-x}\big(b\phi(a)+e^{-k}a\phi(b)\big)\\
&\quad +\frac{1}{M^2} e^{-2x}\bigg\{ e^{-k} \bigg(\Phi(a)\Phi(b)+\frac{1}{\sqrt\gamma}\Phi(a)\phi(b)+\frac{1}{\sqrt\gamma}\phi(a)\Phi(b)\bigg)\\
&\qquad -\frac{\phi(a)}{\sqrt\gamma} \bigg(\Phi(a)+\frac{\phi(a)}{\sqrt\gamma}\bigg)
-e^{-2k}\frac{\phi(b)}{\sqrt\gamma} \bigg(\Phi(b)+\frac{\phi(b)}{\sqrt\gamma}\bigg)\bigg\},
\end{align*}
\begin{align*}
\frac{\partial J}{\partial \theta} &= \bigg(\frac{\partial \gamma}{\partial \theta}\bigg) \bigg[
\frac{1}{M} e^{-x} \big(\phi(a)k_3(k)+e^{-k}\phi(b)k_3(-k)\big)\\
& \qquad +\frac{e^{-2x}}{M^2}\big\{\phi(a)k_1(k)+e^{-k}\phi(b)k_2(-k)\big\}
\bigg\{e^{-k}\bigg(\Phi(b)+\frac{\phi(b)}{\sqrt\gamma}\bigg)-\frac{\phi(a)}{\sqrt\gamma}\bigg\}\\
& \qquad +\frac{e^{-2x}}{M^2}\big\{\phi(a)k_2(k)+e^{-k}\phi(b)k_1(-k)\big\}
\bigg\{\bigg(\Phi(a)+\frac{\phi(a)}{\sqrt\gamma}\bigg)-e^{-k}\frac{\phi(b)}{\sqrt\gamma}\bigg\}\bigg],
\end{align*}
where $k_1,~k_2$ and $k_3$ are defined in Appendix A. Then for some constants $K_i$, 
\begin{align*}
g_t(x,y) & \leq \bigg(\frac{\partial \gamma}{\partial \theta}\bigg)
\bigg\{-\frac{1}{M} \bigg(\frac{1}{2\gamma}\bigg)\big(K_1 b\phi(a)+K_2 a\phi(b)e^{-k}\big)\bigg\} e^{-x}
\end{align*}
and therefore, for some constants $C_i$
\begin{align*}
\bigg| &g_t(x,y)\bigg\{\frac{f_{b_t,d_t}(a_t x,c_t y)-h(x,y)}{A(t)}-\psi(x,y) \bigg\} \bigg|\\
&\qquad \leq C_1 \bigg(\frac{\partial \gamma}{\partial \theta}\bigg) e^{-x} e^{-\frac{2x+k}{1+\rho}}
\cdot \frac{b_t}{2} \exp\bigg\{-\frac{2(1-\rho)x^2+2(1-\rho)kx}{2(1-\rho^2)b_t^2}\bigg\}\\
&\qquad \leq C_2 \phi\bigg(\frac{2x+(3+\rho)b_t^2+k}{\sqrt{2(1+\rho)}b_t}\bigg)
\end{align*}
which implies that (\ref{e:exam_bound_again}) is bounded by an integrable function.
\end{enumerate}

For the general case of $x\rightarrow \infty$ and $y\rightarrow \infty$, the boundness can be obtained. In (\ref{e:g_t}), the first term $\frac{1}{M}\frac{\partial B}{\partial \theta}$ consists of the components; $-e^{-x}\phi\logyx$ and $-e^{-y}\phi\logxy$. In the second term of $g_t(x,y)$, $J(x,y;\theta)^{-1} \big(\frac{\partial J(x,y;\theta)}{\partial \theta}\big)$ is also dominated by $e^{-x}\phi\logyx$ and $e^{-y}\phi\logxy$. Then (\ref{e:exam_bound_again}) is bounded by a function of $\phi(C_1 x,C_2 y)$ for a constant $C_i$, which is integrable.\\

\textbf{B. Derivatives of the Functions}
\[
B(x,y;\theta) =\bigg\{ -e^{-x}\Phi \logyx-e^{-y}\Phi \logxy \bigg\}:
\]
\begin{align*}
\frac{\partial B}{\partial \theta} & = \bigg(\frac{\partial \gamma}{\partial \theta}\bigg) \bigg\{ -e^{-x} \phi \logyx \llogyx \\
& \quad -e^{-y} \phi \logxy \llogxy\bigg\}, \\
\frac{\partial B}{\partial x} & = e^{-x} \Phi \logyx+ \frac{e^{-x}}{\sqrt\gamma} \phi \logyx-\frac{e^{-y}}{\sqrt\gamma}\phi \logxy, \\
\frac{\partial B}{\partial y} & = e^{-y} \Phi \logxy+ \frac{e^{-y}}{\sqrt\gamma} \phi \logxy-\frac{e^{-x}}{\sqrt\gamma}\phi \logyx.\\
\\
\frac{\partial}{\partial \theta}\bigg(\frac{\partial B}{\partial x}\bigg) & = \bigg(\frac{\partial \gamma}{\partial \theta}\bigg)
\bigg\{e^{-x}\phi \logyx \bigg(\frac{1}{8\sqrt \gamma}-\frac{1}{2\sqrt{\gamma^3}}-\frac{y-x}{2\sqrt{\gamma^3}}+\frac{(y-x)^2}{2\gamma\sqrt{\gamma^3}}\bigg)\\
& \quad +e^{-y}\phi \logxy \bigg(\frac{1}{8\sqrt \gamma}+\frac{1}{2\sqrt{\gamma^3}}-\frac{(x-y)^2}{2\gamma\sqrt{\gamma^3}}\bigg) \bigg\}\\
& \doteq \bigg(\frac{\partial \gamma}{\partial \theta}\bigg) \bigg\{e^{-x}\phi \logyx k_1(y-x)+e^{-y}\phi \logxy k_2(x-y) \bigg\}\\
\mbox{where }& k_1(x)=\frac{1}{8\sqrt \gamma}-\frac{1}{2\sqrt{\gamma^3}}-\frac{x}{2\sqrt{\gamma^3}}+\frac{x^2}{2\gamma\sqrt{\gamma^3}}\\
\mbox{and }& k_2(x)=\frac{1}{8\sqrt \gamma}+\frac{1}{2\sqrt{\gamma^3}}-\frac{x^2}{2\gamma\sqrt{\gamma^3}}.
\end{align*}

\begin{align*}
\frac{\partial}{\partial \theta}\bigg(\frac{\partial B}{\partial y}\bigg) & = \bigg(\frac{\partial \gamma}{\partial \theta}\bigg)
\bigg\{e^{-y}\phi \logxy \bigg(\frac{1}{8\sqrt \gamma}-\frac{1}{2\sqrt{\gamma^3}}-\frac{x-y}{2\sqrt{\gamma^3}}+\frac{(x-y)^2}{2\gamma\sqrt{\gamma^3}}\bigg)\\
& \quad +e^{-x}\phi \logyx \bigg(\frac{1}{8\sqrt \gamma}+\frac{1}{2\sqrt{\gamma^3}}-\frac{(y-x)^2}{2\gamma\sqrt{\gamma^3}}\bigg) \bigg\}\\
& \doteq \bigg(\frac{\partial \gamma}{\partial \theta}\bigg) \bigg\{e^{-y}\phi \logxy k_1(x-y)+e^{-x}\phi \logyx k_2(y-x) \bigg\}\\
\end{align*}

\begin{align*}
\frac{\partial^2 B}{\partial x \partial y} & = e^{-x} \phi \logyx \bigg(\frac{1}{2\sqrt\gamma}-\frac{y-x}{\sqrt{\gamma^3}}\bigg)
+e^{-y} \phi \logxy \bigg(\frac{1}{2\sqrt\gamma}-\frac{x-y}{\sqrt{\gamma^3}}\bigg), \\
\frac{\partial}{\partial \theta}\bigg(\frac{\partial^2 B}{\partial x \partial y}\bigg) & = \bigg(\frac{\partial \gamma}{\partial \theta}\bigg)
\bigg[ e^{-x} \phi \logyx \times \\
& \qquad \bigg\{-\frac{1}{16\sqrt\gamma}-\frac{1}{4\sqrt{\gamma^3}}+\bigg(\frac{1}{8\sqrt{\gamma^3}}+\frac{3}{2\gamma\sqrt{\gamma^3}}\bigg)(y-x)
+\frac{(y-x)^2}{4\gamma\sqrt{\gamma^3}}-\frac{(y-x)^3}{2\gamma^2\sqrt{\gamma^3}}\bigg\}\\
& \quad + e^{-y} \phi \logxy \times \\
& \qquad \bigg\{-\frac{1}{16\sqrt\gamma}-\frac{1}{4\sqrt{\gamma^3}}+\bigg(\frac{1}{8\sqrt{\gamma^3}}+\frac{3}{2\gamma\sqrt{\gamma^3}}\bigg)(x-y)
+\frac{(x-y)^2}{4\gamma\sqrt{\gamma^3}}-\frac{(x-y)^3}{2\gamma^2\sqrt{\gamma^3}}\bigg\}\bigg]\\
& \doteq \bigg(\frac{\partial \gamma}{\partial \theta}\bigg)
\bigg[e^{-x} \phi \logyx k_3(y-x)+e^{-y} \phi \logxy k_3(x-y)\bigg],\\
\mbox{where } & k_3(x)=-\frac{1}{16\sqrt\gamma}-\frac{1}{4\sqrt{\gamma^3}}+\bigg(\frac{1}{8\sqrt{\gamma^3}}+\frac{3}{2\gamma\sqrt{\gamma^3}}\bigg)x
+\frac{x^2}{4\gamma\sqrt{\gamma^3}}-\frac{x^3}{2\gamma^2\sqrt{\gamma^3}}.
\end{align*}
Let $a=\frac{\sqrt\gamma}{2}+\frac{y-x}{\sqrt\gamma}$ and $b=\frac{\sqrt\gamma}{2}+\frac{x-y}{\sqrt\gamma}$.
\begin{align*}
J(x,y;\theta) & =\frac{1}{M} \frac{\partial^2 B}{\partial x \partial y}
+\frac{1}{M^2} \frac{\partial B}{\partial x}\cdot\frac{\partial B}{\partial y}\\
& =\frac{1}{M} \bigg\{ e^{-x} \phi(a)b +e^{-y} \phi(b)a \bigg\}\\
& \quad +\frac{1}{M^2} \bigg\{ e^{-x}e^{-y} \bigg(\Phi(a)\Phi(b)+\frac{1}{\sqrt\gamma}\Phi(a)\phi(b)+\frac{1}{\sqrt\gamma}\phi(a)\Phi(b)\bigg)\\
& \qquad -e^{-2x} \frac{\phi(a)}{\sqrt\gamma} \bigg(\Phi(a)+\frac{\phi(a)}{\sqrt\gamma}\bigg)
-e^{-2y} \frac{\phi(b)}{\sqrt\gamma} \bigg(\Phi(b)+\frac{\phi(b)}{\sqrt\gamma}\bigg)\bigg\}.
\end{align*}

\begin{align*}
\frac{\partial J}{\partial \theta} & =\frac{1}{M} \frac{\partial}{\partial \theta}\Big(\frac{\partial^2 B}{\partial x \partial y}\Big) +\frac{1}{M^2} \frac{\partial}{\partial \theta}\Big(\frac{\partial B}{\partial x}\Big) \cdot \frac{\partial B}{\partial y}
+\frac{1}{M^2} \frac{\partial B}{\partial x}\cdot\frac{\partial}{\partial \theta}\Big(\frac{\partial B}{\partial y}\Big)\\
& = \bigg(\frac{\partial \gamma}{\partial \theta}\bigg) \bigg[
\frac{1}{M} \bigg\{e^{-x}\phi(a) k_3(y-x)+e^{-y} \phi(b) k_3(x-y)\bigg\}\\
& \quad +\frac{1}{M^2} \bigg\{e^{-x}\phi(a)k_1(y-x)+e^{-y}\phi(b)k_2(x-y)\bigg\}
\bigg\{e^{-y}\bigg(\Phi(b)+\frac{\phi(b)}{\sqrt\gamma}\bigg)-e^{-x}\frac{\phi(a)}{\sqrt\gamma}\bigg\}\\
& \quad +\frac{1}{M^2} \bigg\{e^{-y}\phi(b)k_1(x-y)+e^{-x}\phi(a)k_2(y-x)\bigg\}
\bigg\{e^{-x}\bigg(\Phi(a)+\frac{\phi(a)}{\sqrt\gamma}\bigg)-e^{-y}\frac{\phi(b)}{\sqrt\gamma}\bigg\}
\bigg].
\end{align*}
\\

\textbf{C. Proof of Theorem 1}

WLOG, assume $w_K=w_{ij}\big((\bm{s}_i,\bm{s}_j)\big)=0 \quad \forall \bm{s} \in R_n^c$.
\begin{align*}
\sigma_K^2 &=\sum_i\sum_{j>i} \sum_p\sum_{q>p} w_{ij}\big(\lambda_n (\bm{x}_i,\bm{x}_j)\big) w_{pq}\big(\lambda_n (\bm{x}_p,\bm{x}_q)\big) \sigma \big(\lambda_n (\bm{x}_i,\bm{x}_j),\lambda_n (\bm{x}_p,\bm{x}_q)\big)\\
&\equiv \sum_i\sum_{j>i} \sum_p\sum_{q>p} h_K(\mathbf{X}_{ij},\mathbf{X}_{pq}), \quad \mathbf{X}_{ij}=(\bm{x}_i,\bm{x}_j)
\end{align*}
Assume that $f(\bm{x}_i,\bm{x}_j)=f(\bm{x}_i)f(\bm{x}_j) \in [m_f,M_f]$ where $m_f$ and $M_f$ are constants.
\begin{align*}
\bigg| &\frac{\int \int w_{ij}\big(\lambda_n (\bm{x}_i,\bm{x}_j)\big)w_{pq}\big(\lambda_n (\bm{x}_i,\bm{x}_j)+\bm{h}\big)f^2(\bm{x}_i,\bm{x}_j)d\bm{x}_i d\bm{x}_j}{\int \int w_{ij}^2\big(\lambda_n (\bm{x}_i,\bm{x}_j)\big)f(\bm{x}_i, \bm{x}_j)d\bm{x}_i d\bm{x}_j} \bigg|\\
& \qquad \leq \frac{M_f^2 \int \int w_{ij}\big(\lambda_n (\bm{x}_i,\bm{x}_j)\big) w_{pq}\big(\lambda_n (\bm{x}_i,\bm{x}_j)+\bm{h}\big) d\bm{x}_i d\bm{x}_j}
{m_f \int \int w_{ij}^2\big(\lambda_n (\bm{x}_i,\bm{x}_j)\big) d\bm{x}_i d\bm{x}_j} \\
& \qquad \leq \bigg(\frac{M_f^2}{m_f}\bigg) \sqrt{\frac{\int \int w_{pq}^2\big(\lambda_n (\bm{x}_i,\bm{x}_j)+\bm{h}\big) d\bm{x}_i d\bm{x}_j}{\int \int w_{ij}^2\big(\lambda_n (\bm{x}_i,\bm{x}_j)\big) d\bm{x}_i d\bm{x}_j}}
\mbox{ (by C-S inequality) }\\
& \qquad \leq \frac{M_f^2}{m_f}<\infty.
\end{align*}
\begin{align*}
E\sigma_K^2 &=K(K-1)E w_K(\lambda_n \mathbf{X}_{ij}) w_K(\lambda_n \mathbf{X}_{pq})\sigma \big(\lambda_n (\mathbf{X}_{ij}-\mathbf{X}_{pq})\big)\\
& \quad +KE w_K(\lambda_n \mathbf{X}_{ij})^2\sigma(\mathbf{0})\\
&= \frac{n(n-1)(n-2)(n-3)}{4} E w_K(\lambda_n \mathbf{X}_{ij}) w_K(\lambda_n \mathbf{X}_{pq})\sigma \big(\lambda_n (\mathbf{X}_{ij}-\mathbf{X}_{pq})\big)\\
& \quad +n(n-1)(n-2) E w_K(\lambda_n \mathbf{X}_{ij}) w_K(\lambda_n \mathbf{X}_{iq})\sigma \big(\lambda_n (\mathbf{X}_{ij}-\mathbf{X}_{iq})\big)\\
& \quad +\frac{n(n-1)}{2} E w_K(\lambda_n \mathbf{X}_{ij})^2 \sigma(\mathbf{0})\\
&= \frac{n(n-1)(n-2)(n-3)}{4} \lambda_n^{-2d} \int \sigma(\bm{h}) \int w_{ij}(\lambda_n \mathbf{X}_{ij})w_{pq}(\lambda_n \mathbf{X}_{ij}+\bm{h}) \times\\
&\qquad f(\mathbf{X}_{ij})f(\mathbf{X}_{ij}+\lambda_n^{-1}\bm{h}) d\mathbf{X}_{ij} d\bm{h}\\
&\quad +n(n-1)(n-2) \lambda_n^{-d} \int \sigma\big(({0,\bm{h}})\big) \times\\
&\qquad \int w_{ij}(\lambda_n \mathbf{X}_{ij})w_{iq}\big(\lambda_n \mathbf{X}_{ij}+({0,\bm{h}})\big)f(\mathbf{X}_{ij})f\big(\mathbf{X}_{ij}+\lambda_n^{-1}({0,\bm{h}})\big) d\mathbf{X}_{ij} d\bm{h}\\
&\quad +\frac{n(n-1)}{2} E w_K(\lambda_n \mathbf{X}_{ij})^2 \sigma(\mathbf{0})\\
&\longrightarrow Kn^2 \lambda_n^{-2d} E w_K^2(\lambda_n \mathbf{X}_{1})\int\sigma(\bm{h})Q_1(\bm{h})d\bm{h}\\
&\qquad +Kn \lambda_n^{-d} E w_K^2(\lambda_n \mathbf{X}_{1})\int\sigma\big((0,\bm{h})\big)Q_2(\bm{h})d\bm{h} +K E w_K(\lambda_n \mathbf{X}_{1})^2 \sigma(\mathbf{0})\\
&\quad = (Ks_{1k}^2) \bigg(\sigma(\mathbf{0})+C_1\int\sigma\big((0,\bm{h})\big)Q_2(\bm{h})d\bm{h}+C_1^2\int\sigma(\bm{h})Q_1(\bm{h})d\bm{h} \bigg),
\end{align*}
as $n\rightarrow \infty~(K\rightarrow\infty)$, by (A$^\prime$4), (A$^\prime$5) and dominated convergence theorem.
\begin{align*}
\sigma_K^2 &= \sum_{a=1}^K \sum_{b=1}^K h_K(\mathbf{X}_{a},\mathbf{X}_{b}),\\
h_{1K}(\bm{x}) &= E h_K (\bm{x},\mathbf{X}_1), ~\bm{x} \in \mathbb{R}^{2d}\\
\mbox{(by Eq. (5.6) in \cite{lahiri:2003})}&\\
\sigma_K^2-E\sigma_K^2 &=\sum_{a=1}^K \big[ h_K(\mathbf{X}_{a},\mathbf{X}_{a})-Eh_K(\mathbf{X}_{1},\mathbf{X}_{1}) \big]\\
&\quad +\sum_{b=1}^{K-1}(K-b) \big[h_{1K}(\mathbf{X}_{b})-E h_{1K}(\mathbf{X}_{b}) \big]\\
&\quad +\sum_{a=2}^K \sum_{b=1}^{a-1} \big[h_K(\mathbf{X}_{a},\mathbf{X}_{b})-Eh_K(\mathbf{X}_{b},\mathbf{X}_{1}) \big]\\
&\doteq D_{1K}+D_{2K}+D_{3K} \\
\big|E h_K(\bm{x},\mathbf{X}_1)^r \big| &=\bigg| \int \int w_K(\lambda_n \bm{x})^r w_K(\lambda_n \bm{s})^r \sigma^r(\lambda_n \bm{x},\lambda_n \bm{s})f(\bm{s})d\bm{s} \bigg|\\
& \leq (M_k^2)^r M_f \lambda_n^{-2d} \int |\sigma(\bm{s})|^r d\bm{s}, \mbox{ by (A'1).}\\
\big|E h_K(\mathbf{X}_1,\mathbf{X}_2)^r \big| & \leq E \big|E h_K(\mathbf{X}_1,\mathbf{X}_2)^r|\mathbf{X}_1) \big|
\leq C\big(M_f, \sigma(\cdot)\big) M_k^{2r} \lambda_n^{-2d}.
\end{align*}
Then
\begin{align*}
\sum_{K=1}^\infty E\big(\sigma_K^2-E\sigma_K^2 \big)^4 &\big/ \big(K^2 \lambda_n^{-2d}E w_K^2(\lambda_n \mathbf{X}_{1}) \big)^4\\
& \leq C(M_f,\sigma(\cdot),C_1) \sum_{K=1}^\infty \bigg(\frac{M_k^2}{E w_K^2(\lambda_n \mathbf{X}_{1})}\bigg)^4 \bigg(\frac{K^6 \lambda_n^{-8d}}{K^8 \lambda_n^{-8d}} \bigg)\\
&= C(M_f,\sigma(\cdot),C_1) \sum_{K=1}^\infty (\gamma_{1K}^2)^4 \bigg(\frac{K^6 \lambda_n^{-8d}}{K^8 \lambda_n^{-8d}} \bigg)\\
&= C(M_f,\sigma(\cdot),C_1) \sum_{K=1}^\infty \frac{(\gamma_{1K}^2)^4}{K^2} <\infty, \quad \mbox{by (A'6)}
\end{align*}
since
\begin{align*}
E D_{1n}^4 &\leq C\big\{ K E h_K(\mathbf{X}_1,\mathbf{X}_1)^4+K^2 \big(E h_K(\mathbf{X}_1,\mathbf{X}_1)^2\big)^2\big\}\\
&\leq C \sigma(\mathbf{0})^4 K^2 (s_{1K}^2 \gamma_{1k}^2)^4 \lambda_n^{-4d}
\leq C \sigma(\mathbf{0})^4 K E w_K(\lambda_n \mathbf{X}_1)^8,\\
E D_{2n}^4 &\leq C\big[ \sum_{b=1}^K (K-b)^4 E h_{1K}(\mathbf{X}_1)^4+\big\{\sum_{a=1}^K (K-a)^2 E h_{1K}(\mathbf{X}_1)^2\big\}^2\big]\\
&\leq C(M_f,\sigma(\cdot)) K^6 M_k^8 \lambda_n^{-8d},\\
E D_{3n}^4 &\leq C K \sum_{a=2}^K E \big\{\sum_{b=1}^{a-1}\big(h_K(\mathbf{X}_a,\mathbf{X}_b)-h_{1K}(\mathbf{X}_b)\big)\big\}^4\\
&\leq C K \sum_{a=2}^K \bigg[E\big\{ (a-1)E \big[ (h_K(\mathbf{X}_a,\mathbf{X}_1)-h_{1K}(\mathbf{X}_1))^4|\mathbf{X}_a\big]\\
& \qquad +\big((a-1)E[(h_K(\mathbf{X}_a,\mathbf{X}_1)-h_{1K}(\mathbf{X}_1))^2|\mathbf{X}_a] \big)^2\big\}\\
& \qquad +E \big[(a-1)(h_{1K}(\mathbf{X}_a)-Eh_{1K}(\mathbf{X}_1))\big]^4\bigg]\\
&\leq C(M_f,\sigma(\cdot)) M_k^8 [K^3 \lambda_n^{-2d}+K^4 \lambda_n^{-4d}+K^6 \lambda_n^{-8d}]
\end{align*}
(see details in Eq. (5.7)-(5.9), \cite{lahiri:2003}).
It follows the analogous result by Lemma 5.2 (i) in \cite{lahiri:2003}. If $n/\lambda_n^d\rightarrow C_1 \in (0,\infty)$ and (A$^\prime$1), (A$^\prime$4) and (A$^\prime$5) hold, then
\begin{align*}
\big( K\cdot E w_K^2(\lambda_n \mathbf{X}_{1}) & \big)^{-1}\sigma_K^2 \\
&\rightarrow \bigg(\sigma(\mathbf{0})+C_1\int\sigma\big((0,\bm{h})\big)Q_2(\bm{h})d\bm{h}+C_1^2\int\sigma(\bm{h})Q_1(\bm{h})d\bm{h} \bigg).
\end{align*}
Let $\xi_k \triangleq \xi_k(\bm{s}^k)=Z_k(\bm{s}^k)-E Z_k(\bm{s}^k)$. Define for $c>0$,
\begin{eqnarray*}
\eta_k=\xi_k I(|\xi_k|\leq c)-E\xi_0 I(|\xi_0|\leq c)\\
\gamma_k=\xi_k I(|\xi_k|> c)-E\xi_0 I(|\xi_0|> c)
\end{eqnarray*}
where $\xi_0=\xi_k(\mathbf{0})$. Let $S_K^{1*}=\sum_{k=1}^K w_k \eta_k,~S_K^{2*}=\sum_{k=1}^K w_k \gamma_k,$ and
\begin{align*}
\sigma_1^*(\bm{x};c) &=Cov(\xi_k I(|\xi_k|\leq c), \xi_0 I(|\xi_0|\leq c)),\\
\sigma_2^*(\bm{x};c) &=Cov(\xi_k I(|\xi_k|> c),\xi_0 I(|\xi_0|> c)).
\end{align*}
We separate the sum of centered processes into two parts,
\begin{align*}
S_K &\equiv \sum_{k=1}^K w_k(\bm{s}^k)\xi_k(\bm{s}^k)\\
&= \sum_{k=1}^K w_k\eta_k +\sum_{k=1}^K w_k\gamma_k = S_K^{1*}+S_K^{2*}.
\end{align*}
By the moment condition on $Z_k(\cdot)$ and the strong mixing condition,
\begin{align}
\max_{j=1,2} \int \int & |\sigma_j^*(\bm{x};c)|d\bm{x} \leq \int \int \big(E|\xi_k(\mathbf{0})^{2+\delta}|\big)^{2/(2+\delta)} \alpha(|\bm{x}|;1)^{\delta/(2+\delta)}d\bm{x} \notag \\
& \leq C\big(d,\delta,E|\xi_k(\mathbf{0})|^{2+\delta},\beta(1)\big)\int_0^\infty {t}^{d-1}\alpha_1({t})^{\delta/(2+\delta)}d{t} <\infty \notag\\
\Rightarrow \int \int & |\sigma_j^*(\bm{x};c)|d\bm{x} <\infty \quad \forall c>0, j=1,2. \label{e:asympVar1}
\end{align}
Since $|Q_1(\bm{x})|\leq 1$, we obtain that for all $\bm{x}$ and $c$,
\begin{align}
\bigg| &\int\int \sigma(\bm{x})Q_1(\bm{x})d\bm{x} -\int\int \sigma_1^*(\bm{x};c)Q_1(\bm{x})d\bm{x} \bigg| \notag\\
& \leq \int\int \bigg\{ | Cov\big( \xi_k(\bm{x})I(|\xi_k(\bm{x})|> c),\xi_k(\mathbf{0}) \big)| \notag\\
&\qquad +|Cov \big( \xi_k(\bm{x}) I(|\xi_k(\bm{x})|\leq c), \xi_k(\mathbf{0})I(|\xi_k(\mathbf{0})|> c) \big)| \bigg\}d\bm{x} \notag\\
& \leq C(d)\big( E|\xi_k(\mathbf{0})|^{2+\delta} \big)^{\frac{1}{2+\delta}} \big( E|\xi_k(\mathbf{0})I(|\xi_k(\mathbf{0})|> c)|^{2+\delta} \big)^{\frac{1}{2+\delta}} \notag\\
&\qquad \times \int_0^\infty {t}^{d-1}\alpha_1({t})^{\delta/(2+\delta)}d{t} \notag\\
& \longrightarrow 0 \quad \mbox{ as } c \rightarrow \infty \notag\\
\Rightarrow &\int\int \sigma(\bm{x})Q_1(\bm{x})d\bm{x} -\int\int \sigma_1^*(\bm{x};c)Q_1(\bm{x})d\bm{x}=o(1). \label{e:asympVar2}
\end{align}
The similar one can be applied to the form with $Q_2$.
\begin{equation} \label{e:asympVar3}
P\bigg(\lim_{K\rightarrow\infty} \big[ES_K^{2*}(c)-\tilde\sigma_{2,K}^2(c)\big]\big/(K s_{1K}^2)=0 \bigg)=1
\end{equation}
where
\begin{align*}
\tilde\sigma_{2,K}^2(c) & \equiv K(K-1)Ew_K(\lambda_n \mathbf{X}_1)w_K(\lambda_n \mathbf{X}_2)\sigma_2^*\big(\lambda_n (\mathbf{X}_1-\mathbf{X}_2);c\big)\\
&\qquad +K Ew_K^2(\lambda_n \mathbf{X}_1)\sigma_2^*(\mathbf{0};c).
\end{align*}
From the previous proof of the asymptotic variance, we can obtain the result such that for any $c>0$,
\begin{align*}
\tilde\sigma_{2,K}^2(c) &=(K s_{1K}^2) \bigg\{ C_1^2 \int\sigma_2^*(\bm{h};c)Q_1(\bm{h})d\bm{h}+ C_1 \int\sigma_2^*\big((0,h);c\big)Q_2({h})d{h}\bigg\}\\
&\qquad +K s_{1K}^2 \sigma_2^*(\mathbf{0};c)
\end{align*}
as $n\rightarrow \infty$. Since $|\int\sigma_2^*(\bm{h};c)Q_1(\bm{h})d\bm{h}|+|\int \sigma_2^*\big((0,h);c\big)Q_2({h})d{h}|+|\sigma_2^*(\mathbf{0};c)|=o(1)$ as $c\rightarrow \infty$, then by (\ref{e:asympVar1}),(\ref{e:asympVar2}) and (\ref{e:asympVar3}),
\begin{eqnarray*}
P\bigg(\lim_{c\rightarrow\infty}\limsup_{n\rightarrow\infty} ES_K^{2*}(c)\big/(K s_{1K}^2)=0 \bigg)=1.
\end{eqnarray*}

Now we apply a classical Bernstein blocking technique for the proof of asymptotic normality. Notations for the blocking technique of Bernstein are same with those of Lahiri. Let $\{\lambda_{1n}\}$ and $\{\lambda_{2n}\}$ be two sequences satisfying the condition (A$^\prime$6) and $\{\lambda_{3n}\}=\{\lambda_{1n}\}+\{\lambda_{2n}\}$. Then the partition of the region $R_n$ is denoted by
\begin{eqnarray*}
\Gamma_n(l;\mathbf{\epsilon})\equiv I_1(\epsilon_1)\times \cdots I_d(\epsilon_d), \quad \mathbf{\epsilon}=(\epsilon_1,\cdots,\epsilon_d)^\prime \in \{1,2\}^d,
\end{eqnarray*}
where $I_j(\epsilon_j)=(l_j\lambda_{3n},l_j\lambda_{3n}+\lambda_{1n}]$, if $\epsilon_j=1$ and $I_j(\epsilon_j)=(l_j\lambda_{3n}+\lambda_{1n},(l_j+1)\lambda_{3n},]$, if $\epsilon_j=2$. Note that with $q(\mathbf{\epsilon})\equiv[\{1\leq j \leq d:\epsilon_j=1\}]$,
\begin{eqnarray*}
|\Gamma_n(l;\mathbf{\epsilon})|=\lambda_{1n}^{q(\mathbf{\epsilon})}\lambda_{2n}^{d-q(\mathbf{\epsilon})}
\end{eqnarray*}
for all $l$ and $\mathbf{\epsilon}$. Let $\mathbf{\epsilon}_0=(1,\cdots,1)^\prime$. Then
\begin{eqnarray*}
|\Gamma_n(l;\mathbf{\epsilon})|=o(|\Gamma_n(l;\mathbf{\epsilon}_0)|).
\end{eqnarray*}
Let $L_{1n}=\{l:\Gamma_n(l;\mathbf{0})\subset R_n\}$ be the index set of all hypercubes $\Gamma_n(l;\mathbf{0})$ that are contained in $R_n$, and let $L_{2n}=\{l:\Gamma_n(l;\mathbf{0})\cap R_n\neq0, \Gamma_n(l;\mathbf{0})\cap R_n^c\neq \emptyset\}$ be the index set of boundary hypercubes. With the notation above, $S_K^{1*}$ can be separated into the sum of big blocks and small blocks and the sum of remaining variables. Here we consider only the case that station elements $i$ and $j$ are in the same block. If sums of pair whose elements are in different block, the joint probability of exceeding over the threshold would be zero as the sampling region is growing. Thus as $n\rightarrow \infty$, sums of pair would converge to 0 and it could be negligible in consideration of our sum of processes.
\begin{align*}
S_K^{1*}/\sigma_K &=\sum_{k=1}^K w_k(\bm{s}^k)\eta_k(\bm{s}^k)\big/\sigma_K\\
&=\sum_{l\in L_{1n}}S_K^{1*}(l;\epsilon_0)+\sum_{\epsilon\neq\epsilon_0}\sum_{l\in L_{1n}}S_K^{1*}(l;\epsilon)+\sum_{l\in L_{2n}}S_K^{1*}(l;\mathbf{0})\\
&=\sum_{q=1}^{|L_{1n}|}\sum_{k \in J_q} w_k \eta_k/\sigma_K +\sum_{q=1}^{|L_{1n}|}\sum_{k \in H_q} w_k \eta_k/\sigma_K +\sum_{k \in L_{2n}} w_k \eta_k/\sigma_K\\
&\triangleq \sum_{q=1}^{|L_{1n}|} S_{1Kq}^\prime+\sum_{q=1}^{|L_{1n}|} S_{2Kq}^\prime+\sum_{k \in L_{2n}} w_k \eta_k/\sigma_K\\
&=S_{1K}^\prime+S_{2K}^\prime+S_{3K}^\prime\\
&\quad \mbox{(big blocks + little blocks + leftover)}
\end{align*}
where $\sigma_K^2=Var\big(\sum_{k=1}^K w_K(\bm{s}_k)\xi_k(\bm{s}_k)\big), S_{1Kq}^\prime=\sum_{k \in J_q} w_k\eta_k/\sigma_K$ and $S_{2Kq}^\prime=\sum_{k \in H_q} w_k\eta_k/\sigma_K$.

Two big blocks $\Gamma(l_1;\epsilon_0)$ and $\Gamma(l_2;\epsilon_0)$ are separated by the distance
\begin{eqnarray*}
d(\Gamma(l_1;\epsilon_0), \Gamma(l_2;\epsilon_0)) \geq [(|l_1-l_2|-d)_+\lambda_{3n}]+\lambda_{2n}.
\end{eqnarray*}
By the strong mixing condition,
\begin{align*}
\bigg|E\exp(itS_{1K}^\prime)&-\prod_{l\in L_{1n}} E\exp\big(it S_K(l;\epsilon_0)\big)\bigg|
\leq C |L_{1n}|\alpha(\lambda_{2n};\lambda_n^d).
\end{align*}
Therefore the asymptotic behavior can be shown with the independence of $S_{1Kq}^\prime$. Using Lemma A.1 in \cite{lahiri:2003}, we show that with probability one,
\begin{align}
\sum_{q=1} &E{S_{1Kq}^\prime}^4 \sigma_K^4=o([K^2 \lambda_n^{-2d}s_{1K}^2]^2),\label{e:var1}\\
&Var(S_{2K}^\prime \sigma_K)=o(K^2 \lambda_n^{-2d}s_{1K}^2), \label{e:var2}\\
&Var(S_{3K}^\prime \sigma_K)=o(K^2 \lambda_n^{-2d}s_{1K}^2). \label{e:var3}
\end{align}
Now we have to show that
\begin{equation} \label{e:asympVar}
\sum_{q=1} E{S_{1Kq}^\prime}^2 \sigma_K^2-\sigma_K^2=o(K^2 \lambda_n^{-2d}s_{1K}^2).
\end{equation}
To prove above equation, we use Lemma 5.1 in \cite{lahiri:2003} and (\ref{e:var1})-(\ref{e:var3}).
\begin{align*}
\bigg|&\sum_{q=1} E{S_{1Kq}^\prime}^2 \sigma_K^2-\sigma_K^2 \bigg|\\
& \leq \bigg| \sum_{q=1} E{S_{1Kq}^\prime}^2 \sigma_K^2-E(S_{1K}^\prime \sigma_K)^2 \bigg|\\
&\qquad +2\sigma_K^2 \big(E(S_{2K}^\prime+S_{3K}^\prime)^2\big)^{1/2}E({S_{1K}^\prime}^2)^{1/2}+ E(S_{2K}^\prime+S_{3K}^\prime)^2 \sigma_K^2\\
& \leq C\bigg[\sum C_0^2 M_n^2 (\lambda_{1n}^{2d}n^2\lambda_n^{-2d}+\log n)^2 \alpha([(|l_1-l_2|-d)_+\lambda_{3n}]+\lambda_{2n};\lambda_{1n}^d)\bigg]\\
&\qquad +o(K^2 \lambda_n^{-2d}s_{1K}^2)\\
& \leq C(d,C_0) M_n^2 (\lambda_{1n}^{2d}n^2\lambda_n^{-2d}+\log n)^2 (\lambda_n/\lambda_{3n})^{2d} \times\\
&\qquad \bigg(\alpha(\lambda_{2n};\lambda_{1n}^d)+\sum_{k=1}^{\lambda_n/\lambda_{3n}}k^{d-1}
\alpha(k\lambda_{3n}+\lambda_{2n};\lambda_{1n}^d)\bigg)+o(K^2 \lambda_n^{-2d}s_{1K}^2)\\
&=o(K^2 \lambda_n^{-2d}s_{1K}^2)
\end{align*}
Thus we show that the equation (\ref{e:asympVar}) holds and it is needed only to establish the Lindeberg condition,
\begin{eqnarray*}
\sum_{q=1}^{|L_{1n}|} E(S_{1Kq}^\prime)^2 I_{(|S_{1Kq}^\prime|>\epsilon)}\longrightarrow 0, \quad \mbox{as }n\rightarrow\infty.
\end{eqnarray*}
Since we have
\begin{align*}
\sum_{q=1}^{|L_{1n}|} & \int_{|S_{1Kq}^\prime|>\epsilon} |S_{1Kq}^\prime|^{2+\delta}dP
=\sum_{q=1}^{|L_{1n}|}\int_{\big|\sum_{k \in J_q} w_k\eta_k/\sigma_K\big|>\epsilon}\bigg|\sum_{k \in J_q} w_k\eta_k\bigg|^{2+\delta}\big/ \sigma_K^{2+\delta}dP\\
& \leq C|L_{1n}|n^{-d\big(\frac{2+\delta}{2}\big)} \bigg(\frac{1}{\sigma^2}\bigg)^{(2+\delta)/2}\int_{\big|\sum_{k \in J_q} w_k\eta_k/\sigma_K\big|>\epsilon}\bigg|\sum_{k \in J_q} w_k\eta_k\bigg|^{2+\delta}dP\\
& \leq C \bigg(\bigg[\frac{n}{\lambda_{1n}}\bigg]\bigg)^d n^{-d\big(\frac{2+\delta}{2}\big)} \int_{\big|\sum_{k \in J_q} w_k\eta_k/\sigma_K \big|>\epsilon}\bigg|\sum_{k \in J_q} w_k\eta_k\bigg|^{2+\delta}dP\\
& \leq C (n^{\delta/2}\lambda_{1n})^{-d}M_k^2 \int_{\big|\sum_{k \in J_q} w_k\eta_k/\sigma_K \big|>\epsilon}\bigg|\sum_{k \in J_q} \eta_k\bigg|^{2+\delta}dP\\
&\longrightarrow 0 \quad \mbox{as } n\rightarrow \infty,
\end{align*}
this implies that the Lindeberg condition holds.\\

\bibliographystyle{plain}
\bibliography{Thesis}

\end{document}